%% file: main.tex
\documentclass[10pt,journal,english]{IEEEtran} %
\usepackage{babel}
\setlocalecaption{english}{table}{\MakeUppercase{Table}}
\usepackage[babel]{microtype}
\usepackage{csquotes}

\usepackage[T1]{fontenc}
\usepackage[largesc]{newtxtext}
\usepackage{newtxtt}
\usepackage[svgnames]{xcolor}
\usepackage{graphicx}
\usepackage{tikz}
\usepackage{pgfplots}
\usepackage[caption=false]{subfig}
\usepackage{stfloats}

\usepackage{tabularx}
\usepackage{booktabs}

\usepackage{mathtools}
\usepackage{amsfonts}
\usepackage{amsthm}
\usepackage{bm}
\usepackage{siunitx}
\usepackage{dsfont}

\usepackage[backend=biber,url=false,style=ieee,bibstyle=ieee-proc,dashed=false,isbn=false,doi=true]{biblatex}
\renewcommand*{\bibfont}{\small}
\usepackage{hyperref}
\hypersetup{
	pdfauthor={Karl-Ludwig Besser},
	pdftitle={On the Distribution of a Two-Dimensional Random Walk with Restricted Angles},
	pdflang={en-US},
	bookmarksnumbered=true,
	colorlinks=true,
	allcolors=.,
	urlcolor=MidnightBlue,
}
\usepackage[all]{hypcap}

\usepackage[acronyms,nomain,xindy,nowarn]{glossaries}
\makeglossaries
\loadglsentries{acronyms.tex}
\setacronymstyle{long-short}
\glsdisablehyper

\input{setup-colors.tex}
\input{setup-math.tex}
\input{setup-plots.tex}
\input{setup-misc.tex}
\input{setup-notation.tex}

\title{On the Distribution of a Two-Dimensional Random~Walk with Restricted Angles}
\author{Karl-Ludwig Besser, \IEEEmembership{Member, IEEE}
	\thanks{K.-L. Besser is with the Department of Electrical Engineering, Linköping University, 581\,83 Linköping, Sweden (email: karl-ludwig.besser@liu.se).}
	\thanks{The work of K.-L. Besser is supported by Security Link.}
}

\begin{document}
\maketitle

\begin{abstract}\noindent\boldmath
	In this paper, we derive the distribution of a two-dimensional (complex) random walk in which the angle of each step is restricted to a subset of the circle.
	This setting appears in various domains, such as in over-the-air computation in signal processing.
	In particular, we derive the exact joint and marginal distributions for two steps, numerical solutions for a general number of steps, and approximations for a large number of steps.
	Furthermore, we provide an exact characterization of the support for an arbitrary number of steps.
	The results in this work provide a reference for future work involving such problems.
\end{abstract}
\begin{IEEEkeywords}
	Random walk,
	Directional statistics,
	Probability distribution,
	Phase alignment,
	Over-the-air computation.
\end{IEEEkeywords}
\glsresetall

\input{introduction.tex}
\input{two-steps.tex}
\input{general-n.tex}
\input{large-n.tex}
\input{application-example.tex}
\input{conclusion.tex}

\section*{Acknowledgments}
\phantomsection
\addcontentsline{toc}{section}{Acknowledgments}
The author would like to thank Martin Dahl for bringing this problem to the author's attention.

\printbibliography[heading=bibintoc]

\end{document}

%% file: acronyms.tex
\newacronym{5g}{5G}{Fifth-Generation Wireless Systems}
\newacronym{ai}{AI}{artificial intelligence}
\newacronym{ann}{ANN}{artificial neural network}
\newacronym{ask}{ASK}{amplitude-shift keying}
\newacronym[plural={AWGN}, \glsshortpluralkey={AWGN}]{awgn}{AWGN}{additive white Gaussian noise}

\newacronym{bec}{BEC}{binary erasure channel}
\newacronym{ber}{BER}{bit error rate}
\newacronym{bler}{BLER}{block error rate}
\newacronym{bpsk}{BPSK}{binary phase-shift keying}

\newacronym{cdf}{CDF}{cumulative distribution function}
\newacronym{clt}{CLT}{central limit theorem}
\newacronym{cnn}{CNN}{convolutional neural network}
\newacronym{csi}{CSI}{channel state information}
\newacronym{csit}{CSI\babelhyphen{nobreak}T}{channel-state information at the transmitter}
\newacronym{csir}{CSI\babelhyphen{nobreak}R}{channel-state information at the receiver}

\newacronym{ecc}{ECC}{error-correcting code}
\newacronym{ecdf}{ECDF}{empirical distribution function}
\newacronym{ee}{EE}{energy efficiency}

\newacronym{ff}{FF}{feed-forward}
\newacronym{ffnn}{FF-NN}{feed-forward neural network}
\newacronym{fsk}{FSK}{frequency-shift keying}

\newacronym{gm}{GM}{Gaussian mixture}

\newacronym{iid}{i.i.d.}{independent and identically distributed}
\newacronym{il}{IL}{information leakage}
\newacronym{iot}{IoT}{internet of things}

\newacronym{jcas}{JCAS}{joint communications and sensing}

\newacronym{lb}{LB}{lower bound}
\newacronym{los}{LoS}{line-of-sight}

\newacronym{mac}{MAC}{multiple access channel}
\newacronym{map}{MAP}{maximum a posteriori}
\newacronym{mc}{MC}{Monte Carlo}
\newacronym{mimo}{MIMO}{multiple-input multiple-output}
\newacronym{miso}{MISO}{multiple-input single-output}
\newacronym{ml}{ML}{machine learning}
\newacronym{mmwave}{mmWave}{millimeter wave}
\newacronym{mop}{MOP}{multi-objective programming problem}
\newacronym{mrc}{MRC}{maximum ratio combining}
\newacronym{msc}{MSC}{modulation and coding scheme}
\newacronym{mse}{MSE}{mean squared error}

\newacronym{nlos}{NLoS}{non-line-of-sight}
\newacronym{nn}{NN}{neural network}
\newacronym{nve}{NVE}{normalized validation error}

\newacronym{oac}{OAC}{over-the-air computation}
\newacronym{ota}{OTA}{over-the-air}

\newacronym{pdf}{PDF}{probability density function}
\newacronym{pwc}{PWC}{polar wiretap code}

\newacronym{qam}{QAM}{quadrature amplitude modulation}

\newacronym{ra}{RA}{rearrangement algorithm}
\newacronym{rbf}{RBF}{radial basis function}
\newacronym{relu}{ReLU}{rectified linear unit}
\newacronym{ris}{RIS}{reconfigurable intelligent surface}
\newacronym{rnn}{RNN}{recurrent neural network}

\newacronym{sgd}{SGD}{stochastic gradient descent}
\newacronym{sgp}{SGP}{secure goodput}
\newacronym{simo}{SIMO}{single-input multiple-output}
\newacronym{sinr}{SINR}{signal-to-interference-plus-noise ratio}
\newacronym{siso}{SISO}{single-input single-output}
\newacronym{snr}{SNR}{signal-to-noise ratio}
\newacronym{svm}{SVM}{support vector machine}

\newacronym{tvd}{TVD}{total variation distance}

\newacronym{uav}{UAV}{unmanned aerial vehicle}
\newacronym{ub}{UB}{upper bound}
\newacronym{urllc}{URLLC}{ultra-reliable low latency communication}

\newacronym{v2v}{V2V}{vehicle-to-vehicle}
\newacronym{v2x}{V2X}{vehicle-to-everything}

\newacronym{zoc}{ZOC}{zero-outage capacity}

%% file: setup-colors.tex
\definecolor{plot0}{HTML}{004488}
\definecolor{plot1}{HTML}{DDAA33}
\definecolor{plot2}{HTML}{BB5566}
\definecolor{plot3}{HTML}{000000}
\definecolor{plot4}{HTML}{AAAAAA}

\definecolor{plot5}{HTML}{00BB44}

%% file: setup-math.tex
\DeclarePairedDelimiter{\abs}{\vert}{\vert}

\DeclarePairedDelimiter{\ceil}{\lceil}{\rceil}
\DeclarePairedDelimiter{\floor}{\lfloor}{\rfloor}

\newcommand*{\inv}[1]{\ensuremath{#1^{-1}}}
\newcommand*{\avg}[1]{\ensuremath{\bar{#1}}}
\newcommand*{\transpose}[1]{\ensuremath{#1^{\mathsf{T}}}}

\newcommand*{\diff}{\ensuremath{\ \mathrm{d}}}
\newcommand*{\imag}{\ensuremath{\mathrm{j}}}
\newcommand*{\e}{\ensuremath{\mathrm{e}}}

\DeclareMathOperator{\sign}{sign}
\DeclareMathOperator{\cov}{cov}

\newcommand*{\reals}{\ensuremath{\mathds{R}}}
\newcommand*{\complexes}{\ensuremath{\mathds{C}}}

\newcommand*{\expect}[2][]{\ensuremath{\mathbb{E}_{#1}\left[#2\right]}}

\newcommand*{\unif}{\ensuremath{\mathcal{U}}}
\newcommand*{\normaldist}{\ensuremath{\mathcal{N}}}
\newcommand*{\gxsdist}{\ensuremath{\tilde{\chi}}}

\newcommand*{\cdf}{\ensuremath{F}}
\newcommand*{\pdf}{\ensuremath{f}}

\newcommand*{\cdfnormal}{\ensuremath{\Phi}}
\newcommand*{\pdfnormal}{\ensuremath{\phi}}

\newcommand*{\mean}{\ensuremath{\mu}}
\newcommand*{\std}{\ensuremath{\sigma}}
\newcommand*{\var}{\ensuremath{\std^2}}

\NewCommandCopy{\rv}{\bm}
\let\Re\relax
\DeclareMathOperator{\Re}{Re}
\let\Im\relax
\DeclareMathOperator{\Im}{Im}

\NewCommandCopy{\originalmiddle}{\middle}
\def\middle#1{\mathrel{}\originalmiddle#1\mathrel{}}

%% file: setup-plots.tex
\pgfplotsset{compat=newest}
\usepgfplotslibrary{fillbetween}

\pgfplotscreateplotcyclelist{lineplot cycle}{ %
	{plot0, mark=*, thick, mark options=solid},
	{plot1, mark=triangle*, thick, mark options=solid},
	{plot2, mark=square*, thick, mark options=solid},
	{plot3, mark=diamond*, thick, mark options=solid},
	{plot4, mark=pentagon*, thick, mark options=solid},
}

\pgfplotsset{%
	betterplot/.style={
		width=.93\linewidth,
		height=.27\textheight,
		xlabel near ticks,
		ylabel near ticks,
		cycle list name=lineplot cycle,
		mark options=solid,
		xmajorgrids=true,
		xminorgrids=true,
		ymajorgrids=true,
		grid style={line width=.1pt, draw=gray!20},
		major grid style={line width=.25pt,draw=gray!30},
		legend cell align=left,
		legend style = {
			/tikz/every even column/.append style={column sep=0.33cm}
		},
	},
}

%% file: setup-misc.tex
\usetikzlibrary{positioning,calc,external,spy}
\newcommand{\todo}[2][]{\ignorespaces\leavevmode
	\if\relax\detokenize{#1}\relax
	{\color{red}[TODO: #2]}%
	\else
	{\color{red}[TODO (#1): #2]}%
	\fi
}

\definecolor{change}{HTML}{0096b8}

\theoremstyle{plain}%
\newtheorem{thm}{Theorem}
\newtheorem{cor}{Corollary}
\newtheorem{lem}{Lemma}

\theoremstyle{definition}
\newtheorem{prob}{Problem}
\newtheorem*{prob*}{Problem Statement}

\theoremstyle{remark}
\newtheorem{rem}{Remark}
\newtheorem*{rem*}{Remark}

\newtheoremstyle{example}{\topsep}{\topsep}{}{}{\itshape}{.}{ }{}
\theoremstyle{example}
\newtheorem{example}{Example}
\newtheorem*{example*}{Example}

\addto\extrasenglish{%

}

\IfPackageLoadedTF{titlesec}{%
	\ifCLASSOPTIONconference%
	\titlespacing{\section}{0pt}{1.5ex plus 1.5ex minus 0.5ex}{0.7ex plus 1ex minus 0ex} %
	\titlespacing{\subsection}{0pt}{1.5ex plus 1.5ex minus 0.5ex}{0.7ex plus .5ex minus 0ex} %
	\else
	\titlespacing{\section}{0pt}{3.0ex plus 1.5ex minus 1.5ex}{0.7ex plus 1ex minus 0ex} %
	\titlespacing{\subsection}{0pt}{3.5ex plus 1.5ex minus 1.5ex}{0.7ex plus .5ex minus 0ex} %
	\fi
	
	\def\thesubsubsectiondis{\arabic{subsubsection})}
	\def\theparagraphdis{\alph{paragraph})}
	\titleformat{\subsubsection}[runin]{\normalfont\normalsize\itshape}{\thesubsubsectiondis}{.5em}{}[:]
	\titlespacing*{\subsubsection}{\parindent}{0ex plus 0.1ex minus 0.1ex}{1ex}
	\titleformat{\paragraph}[runin]{\normalfont\normalsize\itshape}{\theparagraphdis}{.5em}{}[:]
	\titlespacing*{\paragraph}{2\parindent}{0ex plus 0.1ex minus 0.1ex}{1ex}
}{}

%% file: setup-notation.tex
\newcommand*{\numsteps}{\ensuremath{N}}
\newcommand*{\phase}{\ensuremath{\rv{\varphi}}}
\newcommand*{\phaseresult}{\ensuremath{\rv{\theta}}}
\newcommand*{\maxangle}{\ensuremath{a}}

\newcommand*{\radius}{\ensuremath{\rv{R}}}
\newcommand*{\radiusmin}[1][]{
	\if\relax\detokenize{#1}\relax
		\ensuremath{R_{\textnormal{min}}}
	\else
		\ensuremath{R_{\textnormal{min}, #1}}
	\fi
}

\newcommand*{\support}{\mathcal{S}}
\newcommand*{\supportinner}[1][]{%
	\if\relax\detokenize{#1}\relax
		\ensuremath{\support_{\textnormal{in}}}
	\else
		\ensuremath{\support_{\textnormal{in}, #1}}
	\fi
}
\newcommand*{\supportouter}[1][]{%
	\if\relax\detokenize{#1}\relax
		\ensuremath{\support_{\textnormal{out}}}
	\else
		\ensuremath{\support_{\textnormal{out}, #1}}
	\fi
}

\newcommand*{\X}{\ensuremath{\rv{X}}}
\newcommand*{\Y}{\ensuremath{\rv{Y}}}
\newcommand*{\Z}{\ensuremath{\rv{Z}}}

\newcommand*{\ratioXY}{\ensuremath{\rv{W}}}

\newcommand*{\probfa}{\ensuremath{P_{\textnormal{UR}}}}

%% file: introduction.tex
\section{Introduction and Problem Formulation}\label{sec:introduction}
Many signal processing applications involve the superposition of random phasors, e.g., in wireless communications or optics, making the distribution of the resulting amplitude and phase a basic object of interest.
Mathematically, the sum of $\numsteps$~random phasors corresponds to a random walk in the complex plane with $\numsteps$~steps.
Throughout the following, we consider a two-dimensional random walk with $\numsteps$~unit steps, which can be expressed as a random walk in the complex plane as
\begin{equation}\label{eq:def-random-walk}
	\Z_{\numsteps} = \sum_{i=1}^{\numsteps} \exp\left(\imag\phase_i\right) = \radius_{\numsteps}\exp\left(\imag\phaseresult_{\numsteps}\right)\,.
\end{equation}
In this work, we address the problem of finding the joint and marginal probability distributions of the resulting radius~$\radius_{\numsteps}$ and resulting angle~$\phaseresult_{\numsteps}$ after $\numsteps$~steps.
This problem dates back to Pearson~\cite{Pearson1905} and has been extensively studied in directional statistics~\cite{Mardia1972,Jammalamadaka2001}.
However, traditionally, the problem is posed with the assumption that each step of the random walk is taken in an arbitrary direction, i.e., the angles~$\phase_i$ are uniformly distributed over~$[-\pi, \pi]$.
In contrast to the existing literature, we assume that the phases~$\phase_i$ are independent and uniformly distributed with a maximum angle~${\maxangle\leq\frac{\pi}{2}}$, i.e., ${\phase_i\sim\unif[-\maxangle, \maxangle]}$ with \gls{pdf}
\begin{equation}\label{eq:def-pdf-phases}
	\pdf_{\phase_i}(\theta) = \frac{1}{2\maxangle},\quad -\maxangle\leq \theta \leq \maxangle
\end{equation}
and \gls{cdf}
\begin{equation}\label{eq:def-cdf-phases}
	\cdf_{\phase_i}(\theta) = \frac{\theta+\maxangle}{2\maxangle}\,.
\end{equation}

This general problem has various applications in the area of signal processing and beyond.
It is especially relevant for topics related to phase alignment and the detection of phase misalignment, such as in the context of over-the-air computation~\cite{Dahl2024,Shao2021} and distributed antenna systems~\cite{Larsson2024,Nanzer2021}, e.g., for sensing applications.
As a specific application example, consider an over-the-air computation system~\cite{Sahin2023}, in which $\numsteps$~devices transmit their data simultaneously.
Ideally, the signals of all devices are aligned perfectly and there is no phase offset.
However, due to drift effects, the phase misalignment worsens over time and the devices need to be resynchronized.
As this calibration can be expensive, it should only be done when necessary, i.e., when the phase offset of any of the devices is larger than a tolerated threshold.
Detecting this potential misalignment can be reduced to the problem addressed in this work.
In particular, the devices can transmit pilot symbols periodically, such that the received signal (of the superimposed individual signals) has the same structure as the random walk in~\eqref{eq:def-random-walk}.
Assuming a maximum tolerated phase offset of~$\maxangle$, the results of this work can be used to calculate the likelihood that at least one device exceeds the threshold.
This can in turn be used to decide whether a recalibration is necessary, avoiding unnecessary resynchronizations and making the system more efficient.
This application example is discussed in more detail in \autoref{sec:application-example}.

Besides signal processing applications in wireless communications, similar models can be found in optics.
Optical speckle patterns can be described through the superposition of random phasors, including problems with phases that are not uniformly distributed across the full circle~\cite[Sec.~2.6]{Goodman2020}.
Furthermore,
the considered problem also arises in unrelated fields, such as in the modeling of animal movement in biology~\cite{Codling2008,Marsh1988}.

The problem studied throughout the rest of this paper is summarized as follows.

\begin{prob}\label{prob:problem-formulation}
	What are the joint and marginal distributions of the resulting radius~$\radius_{\numsteps}$ and angle~$\phaseresult_{\numsteps}$ for the random walk with $\numsteps$~steps in~\eqref{eq:def-random-walk}, if the phases~$\phase_i$ of the individual steps are uniformly distributed over the interval~${[-\maxangle, \maxangle]}$ with ${\maxangle\leq\frac{\pi}{2}}$?
\end{prob}

\begin{rem}[{Exact Distribution for ${\numsteps=1}$}]
	For the trivial case of only a single step, i.e., $\Z_1=\exp(\imag\phase_1)$, the resulting distributions are straightforward.
	In particular, the radius~${\radius_1=1}$ is constant and the resulting phase~$\phaseresult_1=\phase_1$ is uniformly distributed according to~\eqref{eq:def-cdf-phases}.
\end{rem}

\subsection{Existing Results and Contributions}
The traditional problem with~${\maxangle=\pi}$ has been extensively studied since Pearson posed the problem.
First, it should be noted that the resulting phase is always uniformly distributed over the full circle, just as in each individual step.
Therefore, the discussion of the traditional problem reduces to the discussion of the distribution of the resulting radius~$\radius_{\numsteps}$.

An approximation for a large number of steps has been noted by Rayleigh~\cite{Rayleigh1905}, where the resulting radius follows a Rayleigh distribution.
An exact solution to the traditional problem has been proposed by Kluyver~\cite{Kluyver1905} who showed that the distribution is given as
\begin{equation*}
	\cdf_{\radius_{\numsteps}}(r) = r \int_{0}^{\infty} J_1(rt) J_0(t)^{\numsteps} \diff{t}\,,
\end{equation*}
where $J_\alpha$ are the Bessel functions of the first kind and order~$\alpha$.
However, while it provides an accurate solution, the integral of a product of Bessel functions can be difficult to compute in general.
Therefore, several numerical solutions have been proposed~\cite{Bennett1948,Simon1985}.
A summary of the results for the traditional problem with a uniform distribution of the phases over the full circle is given in \autoref{tab:summary-literature}.

\begin{table}
	\renewcommand*{\arraystretch}{1.25}
	\centering
	\caption{Summary of existing results for the two-dimensional random walk with phases uniformly distributed over the full circle~(${\maxangle=\pi}$).}
	\label{tab:summary-literature}
	\begin{tabularx}{.97\linewidth}{l|XXX}
		\toprule
		& $\numsteps=2$ & $\numsteps=3$ & $\numsteps\to\infty$\\
		\midrule
		Radius~$\radius_{\numsteps}$ & \cite[Sec.~3.2]{Jammalamadaka2001}\newline(Exact) & \cite{Bennett1948,Simon1985}\newline(Numeric.) & Rayleigh~\cite{Rayleigh1905}\newline(Approx.)\\
		Angle~$\phaseresult_{\numsteps}$ & Uniform & Uniform & Uniform\\
		Joint~$(\radius_{\numsteps}, \phaseresult_{\numsteps})$ & Independent & Independent & Independent \\
		\bottomrule
	\end{tabularx}
\end{table}

While the case of a uniform distribution over the full circle has been studied in detail, there are fewer results for different distributions of the phases~$\phase_i$.
One known result is for the von~Mises distribution~\cite[{Sec.~4.5}]{Mardia1972}, of which the uniform distribution is a special case.
However, the support of the von~Mises distribution is still the full circle.
In contrast, in this work, we derive new results for the distribution of the two-dimensional random walk where the angle of each step is uniformly distributed over a subset of the circle as described above.
Existing work on circular statistics with sector-limited distributions, such as the truncated von~Mises distribution~\cite{FernandezGonzalez2017}, typically focuses on single-angle modeling and inference~\cite{Sakthivel2026,Kanika2016} and does not provide endpoint distributions or support geometry for multi-step random walks with restricted angles, which is the focus of this paper.

The contributions of this work are summarized as follows.
\begin{itemize}
	\item We derive analytical expressions of the joint and marginal distributions of the resulting radius~$\radius_{2}$ and angle~$\phaseresult_{2}$ for the two-dimensional random walk with two unit steps and angles~$\phase_i$ uniformly distributed over the interval~${[-\maxangle, \maxangle]}$. (\autoref{thm:distributions-n2})
	\item For an arbitrary number of steps~$\numsteps$, we characterize the support of~$\Z_{\numsteps}$ (\autoref{thm:support-bound-characterization}), provide a parametrization (\autoref{cor:support-parametrization}) and a condition that allows simplifying its structure (\autoref{lem:condition-support-unique}).
	\item We show how the distributions for a small number of steps~${\numsteps\geq3}$ can be calculated using numerical integration, and we provide an implementation for this. (\autoref{thm:distributions-general-n} and~\cite{BesserGithub})
	\item We derive an approximation of the joint and marginal distributions for a large number of steps~$\numsteps$. (\autoref{thm:distributions-large-n})
\end{itemize}
All results are illustrated with numerical examples and the source code to reproduce all calculations and simulations is made freely available at~\cite{BesserGithub}.
An overview of the results can be found in \autoref{tab:summary-contributions}.
Furthermore, we discuss an example of applying the results to a problem in over-the-air computation in \autoref{sec:application-example}.

\begin{table}
	\renewcommand*{\arraystretch}{1.25}
	\centering
	\caption{Summary of the results derived in this work for the two-dimensional random walk with angles uniformly distributed over the interval~${[-\maxangle, \maxangle]}$, ${\maxangle\leq\frac{\pi}{2}}$.}
	\label{tab:summary-contributions}
	\begin{tabularx}{.97\linewidth}{l|XXX}
		\toprule
		& $\numsteps=2$ & $\numsteps=3$ & $\numsteps\to\infty$\\
		\midrule
		Radius~$\radius_{\numsteps}$ & \eqref{eq:result-cdf-radius-n2}/\eqref{eq:result-pdf-radius-n2} (Exact) & \eqref{eq:result-cdf-radius-general-n} (Numer.) & \eqref{eq:result-cdf-radius-large-n}/\eqref{eq:result-pdf-radius-large-n}\newline(Approx.)\\
		Angle~$\phaseresult_{\numsteps}$ & \eqref{eq:result-cdf-phase-n2}/\eqref{eq:result-pdf-phase-n2} (Exact) & \eqref{eq:result-cdf-phase-general-n-approx}/\eqref{eq:result-pdf-phase-general-n-approx}\newline(Approx.) & \eqref{eq:result-cdf-phase-large-n}/\eqref{eq:result-pdf-phase-large-n}\newline(Approx.)\\
		Joint~$(\radius_{\numsteps}, \phaseresult_{\numsteps})$ & \eqref{eq:result-pdf-joint-n2} (Exact) & \eqref{eq:result-pdf-joint-general-n} (Numer.) & \eqref{eq:result-pdf-joint-large-n} (Approx.)\\
		\bottomrule
	\end{tabularx}
\end{table}

\subsubsection*{Notation}
Random variables are denoted in capital boldface letters, e.g., $\X$, and their realizations in small letters, e.g.,~$x$.
We use $F_{\X}$ and $f_{\X}$ for the probability distribution and its density, respectively.
The expectation is denoted by $\mathbb{E}$ and the probability of an event by $\Pr$.
The uniform distribution on the interval~${[a,b]}$ is denoted as~${\unif[a,b]}$.
The normal distribution with mean~$\mu$ and variance~$\sigma^2$ is denoted as~${\normaldist(\mu, \sigma^2)}$.
The derivative of a function $f$ is denoted by $f'$.

%% file: two-steps.tex
\section{Exact Solution for \texorpdfstring{$\numsteps=2$}{N=2}}\label{sec:two-steps}

For the simplest case of $\numsteps=2$, we derive the following exact solution for the distributions of the resulting radius and angle.

\begin{thm}\label{thm:distributions-n2}
	Consider a two-dimensional random walk with two unit-length steps and uniformly distributed phases in the interval~${[-\maxangle, \maxangle]}$, ${0<\maxangle\leq\frac{\pi}{2}}$.
	The resulting radius~$\radius_2$ and resulting angle~$\phaseresult_2$ are distributed as follows.
	\begin{itemize}
		\item The distribution~$\cdf_{\radius_2}$ of the radius~$\radius_2$ is given as
		\begin{equation}\label{eq:result-cdf-radius-n2}
			\cdf_{\radius_2}(r) = \frac{\left(\maxangle - \arccos\left(\frac{r}{2}\right)\right)^2}{\maxangle^2}\,,
			\quad 2\cos\maxangle \leq r \leq 2\,.
		\end{equation}
		The corresponding density~$\pdf_{\radius_2}$ is given as
		\begin{equation}\label{eq:result-pdf-radius-n2}
			\pdf_{\radius_2}(r) = \frac{2\left(\maxangle - \arccos\left(\frac{r}{2}\right)\right)}{\maxangle^2 \sqrt{4 - r^2}}\,.
		\end{equation}
		\item The distribution~$\cdf_{\phaseresult_2}$ of the angle~$\phaseresult_2$ is given as
		\begin{equation}\label{eq:result-cdf-phase-n2}
			\cdf_{\phaseresult_2}(\theta) = \frac{\maxangle^2 + 2\maxangle \theta - \sign(\theta) \theta^2}{2 \maxangle^2},
			\quad -\maxangle \leq \theta \leq \maxangle\,.
		\end{equation}
		The corresponding density~$\pdf_{\phaseresult_2}$ is given as
		\begin{equation}\label{eq:result-pdf-phase-n2}
			\pdf_{\phaseresult_2}(\theta) = \frac{1}{\maxangle} \left(1-\frac{\abs{\theta}}{\maxangle}\right)\,.
		\end{equation}
		\item The joint density~$\pdf_{\radius_2,\phaseresult_2}$ of the radius~$\radius_2$ and angle~$\phaseresult_2$ is given as
		\begin{equation}\label{eq:result-pdf-joint-n2}
			\pdf_{\radius_2, \phaseresult_2}(r, \theta)
			= \begin{cases}
				\frac{1}{\maxangle^2 \sqrt{4 - r^2}} & 2\cos\left(\maxangle-\abs{\theta}\right) \leq r < 2\\
				0 & \text{otherwise.}
			\end{cases}
		\end{equation}
	\end{itemize}
\end{thm}

\begin{proof}
	In the remainder of this section, we prove \autoref{thm:distributions-n2} in the following steps:
	\begin{enumerate}
		\item Derivation of the distribution (and density) of the angle~$\phaseresult_2$ (\autoref{sub:n2-angle})
		\item Derivation of the conditional probability~$\cdf_{\radius_2 \mid \phaseresult_2}$ (\autoref{sub:n2-conditional-r-g-t})
		\item Derivation of the joint density~$\pdf_{\radius_2, \phaseresult_2}$ (\autoref{sub:n2-joint})
		\item Derivation of the marginal distribution (and density) of the radius~$\radius_2$ (\autoref{sub:n2-radius}). %
	\qedhere
	\end{enumerate}
\end{proof}

\input{n2-angle.tex}
\input{n2-conditional-rgt.tex}
\input{n2-joint.tex}

\input{n2-radius.tex}

%% file: n2-angle.tex
\subsection{Distribution of Angle~\texorpdfstring{$\phaseresult_2$}{θ₂}}\label{sub:n2-angle}

As a first step, we derive the marginal distribution of the resulting angle~$\phaseresult_2$ after two steps.
Since the result is clearly symmetric in~$\phaseresult_2$, we will restrict most of our discussion to the case of ${\phaseresult_2\geq 0}$ to simplify the notation.
An illustration of the problem is given in \autoref{fig:illustration-random-walk-n2}.

From the fact that both steps are of equal length, we have the first important observation about~$\phaseresult_2$ based on the half-angle theorem
\begin{equation}\label{eq:angle-result-sum-n2}
	\phaseresult_2 = \frac{\phaseresult_1 + \phase_2}{2}, \quad \text{with~} (\phaseresult_1, \phase_2) \in [-\maxangle, \maxangle]^2
\end{equation}
or, equivalently, we can express~$\phaseresult_1$ as a function of~$\phase_2$ as
\begin{equation}\label{eq:angle-phase2-linear-phase1-n2}
	\phaseresult_1 = 2\phaseresult_2 - \phase_2, \quad \text{with~} \phase_2\in[-\maxangle, \maxangle]\,.
\end{equation}

\begin{figure}
	\centering
	\input{img/illustration-random-walk-notation.tex}
	\vspace{-.75em}
	\caption{%
		Illustration of the random walk with ${\numsteps=2}$~unit steps.%
	}
	\label{fig:illustration-random-walk-n2}
\end{figure}

Based on the linear connection from~\eqref{eq:angle-phase2-linear-phase1-n2}, we can calculate the density~$\pdf_{\phaseresult_2}$ of the angle after two steps~$\phaseresult_2$ by convolution as
\begin{equation}\label{eq:pdf-angle-n2-convolution}
	\pdf_{\phaseresult_2}(\theta) = \int_{-\maxangle}^{\maxangle} \pdf_{\phaseresult_1}(2\theta-\varphi) \pdf_{\phase_2}(\varphi) \diff{\varphi}\,,
\end{equation}
where we use the fact that $\phase_2$ is independent of~$\phase_1=\phaseresult_1$.

For the uniform distribution of the angles~$\phase_i$ from~\eqref{eq:def-pdf-phases}, $\pdf_{\phaseresult_2}$ is calculated as
\begin{equation}\label{eq:pdf-phase-n2-closed-form}
	\pdf_{\phaseresult_2}(\theta) = \frac{1}{\maxangle} \left(1-\frac{\abs{\theta}}{\maxangle}\right), \quad -\maxangle \leq \theta \leq \maxangle\,,
\end{equation}
which corresponds to a triangular distribution.
The corresponding \gls{cdf} is then given as
\begin{equation*}
	\cdf_{\phaseresult_2}(\theta) = \frac{\maxangle^2 + 2\maxangle \theta - \sign(\theta) \theta^2}{2 \maxangle^2}\,.
\end{equation*}

\begin{example}[Distribution of Angle~$\phaseresult_2$]\label{ex:phase-n2}
	To illustrate and numerically verify the above derivation, we show the \gls{pdf} of~$\phaseresult_2$ for ${\maxangle=0.5}$ in \autoref{fig:pdf-angle-n2}.
	Besides the exact solution from~\eqref{eq:pdf-phase-n2-closed-form}, we show the histogram obtained from \gls{mc} simulations with $10^6$~samples.

	\begin{figure}
		\centering
		\input{img/pdf-angle-n2.tex}
		\vspace{-1em}
		\caption{%
			\Gls{pdf} of the resulting angle~$\phaseresult_2$ for ${\maxangle=0.5}$.
			Besides the exact solution from~\eqref{eq:result-pdf-phase-n2}, a numerical \gls{pdf} obtained through \gls{mc} simulation with $10^6$~samples is shown. (\autoref{ex:phase-n2})}
		\label{fig:pdf-angle-n2}
	\end{figure}
\end{example}

%% file: img/illustration-random-walk-notation.tex
\begin{tikzpicture}%
	\begin{axis}[
		betterplot,
		axis lines=middle,
		axis equal image,
		height=.3\textheight,
		xlabel={$\Re\Z_2$},
		ylabel={$\Im\Z_2$},
		xmin=0,
		xmax=2.2,
		ymin=0,
		ymax=2.2,
		]
		\draw[dashed,blue] (0,0) circle [radius=2];
		\draw[thick,densely dashed] (-60:2) -- (0,0) -- (60:2);

		\draw[->,plot0,very thick] (0,0) -- (45:1);
		\draw[plot0] (0:.6) arc (0:45:.6);
		\node[inner sep=0pt,plot0] at (.45,.12) {$\phaseresult_1$};
		\node[inner sep=0pt,plot0] at (.4,.5) {$1$};

		\draw[->,plot1,very thick] (.7071,.7071) -- (1.6468,1.0491); %
		\draw[plot1] ($(.7071,.7071) + (0:.32)$) arc (0:20:.32);
		\draw[gray,dashed,dash phase=35] (.7071,.7071) -- (1.0271,.7071);
		\node[inner sep=0pt,plot1] (labelPhase2) at (.9,.95) {$\phase_2$};
		\draw (labelPhase2) -- (.95,.75);
		\node[inner sep=0pt,plot1] at (1.2,.97) {$1$};

		\draw[->,plot2,very thick] (0,0) -- (1.6468,1.0491); %
		\draw[plot2] (0:1.45) arc (0:32.5:1.45);
		\node[inner sep=0pt,plot2] at (1.3,.3) {$\phaseresult_2$};
		\node[inner sep=0pt,plot2] at (.9,.45) {$\radius_2$};

		\node[inner sep=0pt] at (.25,.9) {$\frac{\pi}{2}-\maxangle$};
		\draw[] (60:1.1) arc (60:90:1.1);
	\end{axis}
\end{tikzpicture}

%% file: img/pdf-angle-n2.tex
\begin{tikzpicture}%
	\begin{axis}[
		betterplot,
		height=.23\textheight,
		xlabel={Angle~$\phaseresult_2$},
		ylabel={\Gls{pdf} $\pdf_{\phaseresult_2}$},
		xmin=-.5,
		xmax=.5,
		ymin=0,
		ymax=2.1,
		]
		\pgfplotstableread{data/results-mc-n2-a0.500.dat}\tmc

		\addplot+[domain=-.5:.5] {2*(1-2*abs(x))};
		\addlegendentry{Exact~(7)}%

		\addplot+[mark repeat=3,const plot mark right] table[x=angle,y=histAngle] {\tmc};
		\addlegendentry{MC}
	\end{axis}
\end{tikzpicture}

%% file: n2-conditional-rgt.tex
\subsection{Conditional Probability~\texorpdfstring{$\cdf_{\radius_2 \mid \phaseresult_2}$}{F(R₂|θ₂)}}\label{sub:n2-conditional-r-g-t}
As a next step, we derive the conditional probability of the resulting radius~$\radius_2$ given the resulting angle~$\phaseresult_2$ after two steps.

By definition of the random walk, the resulting length after two steps is given as
\begin{align}
	\radius_2
	&= \abs*{\exp(\imag\phaseresult_1) + \exp(\imag\phase_2)}\\
	&= \sqrt{2 \left(1 + \cos(\phase_2-\phaseresult_1)\right)}
	\label{eq:radius-n2}\,.
\end{align}

Combining~\eqref{eq:radius-n2} with relation~\eqref{eq:angle-phase2-linear-phase1-n2}, we can write the conditional probability as
\begin{multline}\label{eq:prob-cond-n2-general}
	\cdf_{\radius_2 \mid \phaseresult_2}(r\mid \theta)
	= \Pr\left(\radius_2 \leq r \mid \phaseresult_2=\theta\right)
	\\= \Pr\left(\cos\left(2(\phase_2 - \theta)\right) \leq \frac{r^2}{2}-1 \middle| \phaseresult_2=\theta\right)\,.
\end{multline}
Due to the ambiguity of the cosine, we find the two boundary points solving ${\cos\left(2(\phase_2 - \theta)\right) = \frac{r^2}{2}-1}$ as
\begin{align}
	\tilde{\phase_2} &= \theta - \frac{\arccos\left(\frac{r^2}{2}-1\right)}{2} = \theta - \arccos\frac{r}{2}\\
	\hat{\phase_2} &= \theta +\frac{\arccos\left(\frac{r^2}{2}-1\right)}{2} = \theta + \arccos\frac{r}{2}\,.
\end{align}
An illustration can be found in \autoref{fig:illustration-event-cosine-less}.
With this, we can rewrite the conditional probability in~\eqref{eq:prob-cond-n2-general} to
\begin{align}
	\cdf_{\radius_2 \mid \phaseresult_2}(r\mid \theta)
	&= 1- \Pr\left(\phase_2\in [\tilde{\phase_2}, \hat{\phase_2}] \middle| \phaseresult_2=\theta\right)\\
	&= 1-\left(\cdf_{\phase_2\mid \phaseresult_2}(\hat{\phase_2}\mid \theta) - \cdf_{\phase_2\mid \phaseresult_2}(\tilde{\phase_2}\mid \theta)\right)
	\label{eq:prob-cond-n2-events}.
\end{align}

In order to evaluate~\eqref{eq:prob-cond-n2-events}, we need to derive the conditional probability~$\cdf_{\phase_2\mid \phaseresult_2}(\varphi\mid\theta)$ first.
For the considered two-step random walk, we have the following relation
\begin{align}
	\cdf_{\phase_2\mid \phaseresult_2}(\varphi\mid\theta)
	&= \Pr\left(\phase_2 \leq \varphi \middle| \phaseresult_2=\theta \right)\\
	&\stackrel{(a)}{=} 1 - \Pr\left(\phaseresult_1 \leq 2\theta-\varphi \middle| \phaseresult_2=\theta\right)\\
	&\stackrel{(b)}{=} 1 - \unif_{[-\maxangle, \maxangle]}\left(2\theta-\varphi\right)\\
	&\stackrel{(c)}{=} \unif_{[2\theta-\maxangle, \maxangle]}\left(\varphi\right)
	\,,
\end{align}
where $(a)$~uses relation~\eqref{eq:angle-phase2-linear-phase1-n2} between the angles, $(b)$~follows from the fact that the angle after the first step~${\phaseresult_1=\phase_1}$ is uniformly distributed over~${[-\maxangle, \maxangle]}$, and $(c)$~follows from appropriately rearranging the expressions for the considered intervals and using the assumption~$\theta\geq 0$.
Hence, the conditional distribution~$\cdf_{\phase_2\mid \phaseresult_2}(\varphi\mid\theta)$ is a uniform distribution over the interval~$[2\theta-\maxangle, \maxangle]$.

\begin{figure}
	\centering
	\input{img/illustration-cosine-event-parts.tex}
	\vspace{-.75em}
	\caption{Illustration of the event for calculating the conditional probability~$\cdf_{\radius_2 \mid \phaseresult_2}$ in~\eqref{eq:prob-cond-n2-general}.}
	\label{fig:illustration-event-cosine-less}
\end{figure}

Thus, with this conditional probability, we can evaluate~\eqref{eq:prob-cond-n2-events} and obtain
\begin{align}
	\cdf_{\radius_2 \mid \phaseresult_2}(r\mid \theta)
	&= 1 - \frac{\arccos\left(\frac{r}{2}\right)}{\maxangle-\theta}
	\label{eq:prob-cond-n2-closed-form}\,,
\end{align}
for ${\radiusmin(\theta) \leq r \leq 2}$ and ${0\leq\theta\leq\maxangle}$ with the minimum radius~$\radiusmin$ being
\begin{equation}
		\radiusmin(\theta) = \sqrt{2\left(1+\cos\left(2(\maxangle-\theta)\right)\right)}\,.
\end{equation}

Due to the symmetry around ${\phaseresult_2=0}$, we can replace $\theta$ by its absolute value~$\abs{\theta}$ to obtain the solution for the full range~${\theta\in[-\maxangle, \maxangle]}$.

\begin{rem}
	The general support of~${(\radius_{\numsteps}, \phaseresult_{\numsteps})}$ and the minimum radius~$\radiusmin$ will later be discussed in detail in \autoref{sub:n3-support}.
\end{rem}

%% file: img/illustration-cosine-event-parts.tex
\begin{tikzpicture}%
	\begin{axis}[
		betterplot,
		height=.21\textheight,
		trig format plots=rad,
		xlabel={Angle~$\phase_2$},
		ylabel={Radius~$\radius_2$},
		domain=-1.4:1.4,
		xmin=-1.4, xmax=1.4,
		ymin=0, ymax=2,
		enlarge y limits=.1,
		xtick={0,.3,-1.4,1.4,2*.3-1.4,1.0227342478134157,-0.4227342478134157},
		xticklabels={$0$, $\theta$,$-\maxangle$, $\maxangle$,$2\theta-\maxangle$,$\hat{\phase_2}$, $\tilde{\phase_2}$},
		ytick={0,1.5,2, 0.9071922428511551},
		yticklabels={$0$, $r$, $2$, $\radiusmin$},
		clip mode=individual,
		]
		\addplot+ {sqrt(2*(1+cos(2*(x-.3))))};
		\addplot[draw=none,fill=gray, fill opacity=.75] coordinates {(-1.4,-1) (-1.4,2.5) (2*.3-1.4,2.5) (2*.3-1.4, -1)};
	\end{axis}
\end{tikzpicture}

%% file: n2-joint.tex
\subsection{Joint Density of Radius and Angle~\texorpdfstring{$(\radius_2, \phaseresult_2)$}{(R₂, θ₂)}}\label{sub:n2-joint}
Combining the conditional probability~$\cdf_{\radius_2 \mid \phaseresult_2}$ from~\eqref{eq:prob-cond-n2-closed-form} with the \gls{pdf} of the angle~$\phaseresult_2$ from~\eqref{eq:pdf-phase-n2-closed-form} yields the joint density as
\begin{align}
	\pdf_{\radius_2, \phaseresult_2}(r, \theta)
	&= \frac{\partial}{\partial r} \cdf_{\radius_2 \mid \phaseresult_2}(r \mid \theta) \pdf_{\phaseresult_2}(\theta)\\
	&= \begin{cases}
		\frac{1}{\maxangle^2 \sqrt{4 - r^2}} & 2\cos\left(\maxangle-\abs{\theta}\right) \leq r \leq 2\\
		0 & \text{otherwise.}
	\end{cases}
	\label{eq:joint-density-n2-deriv}%
\end{align}

An important observation about the joint distribution in~\eqref{eq:joint-density-n2-deriv} is its restricted support.
While it is clear from the general assumptions that the resulting angle always needs to be within~${[-\maxangle, \maxangle]}$, the derived expression shows the exact relation between the resulting angle and radius.
In particular, for small radii (within the support), the range of possible angles is closely restricted around zero and widens as the radius increases.
This shows the strong dependence between~$\radius_2$ and~$\phaseresult_2$, in contrast to the traditional problem with a uniform distribution of the phases over the full circle.

\begin{example}[{Joint Distribution of~${(\radius_2, \phaseresult_2)}$}]\label{ex:joint-n2}
	A numerical illustration of the joint \gls{pdf} is given in \autoref{fig:pdf-joint-n2}, where we show the density together with the boundary of its support for~${\maxangle=0.5}$.
	The support widens with an increasing radius, i.e., for values of~$\radius_2$ close to two, all angles can occur.
	At the same time, the likelihood increases with~$\radius_2$.
	This is due to the fact that in order to have a small resulting radius~$\radius_2$, the angles of the two steps~$\phase_1$ and~$\phase_2$ need to cancel out, i.e., $\phase_2\approx-\phase_1$.
	As this scenario is unlikely, most of the probability mass is concentrated at larger radii.

	\begin{figure}
		\centering
		\input{img/pdf-joint-n2.tex}
		\vspace{-1em}
		\caption{%
			Joint distribution of radius~$\radius_2$ and angle~$\phaseresult_2$ for two steps and maximum angle~${\maxangle=0.5}$.
			In addition, the boundary of the support is highlighted. (\autoref{ex:joint-n2})}
		\label{fig:pdf-joint-n2}
	\end{figure}
\end{example}

%% file: img/pdf-joint-n2.tex
\begin{tikzpicture}
	\begin{axis}[
		width=.79\linewidth,
		height=.22\textheight,
		trig format plots=rad,
		view={0}{90},
		xmin=-.5,
		xmax=.5,
		ymax=2,
		colormap name=viridis,
		colorbar style={xtick=data},
		ylabel={Radius~$\radius_2$},
		xlabel={Angle~$\phaseresult_2$},
		axis on top,
		]
		\pgfplotstableread{data/results-joint-n2-a0.500.dat}\tbl

		\addplot3[surf,faceted color=none] table [x=angle,y=radius,z=pdfJoint]{\tbl};
		\addplot3[ultra thick,plot1,dashed,domain=-.5:.5] (x,{2*cos(.5-abs(x))},0);
	\end{axis}
\end{tikzpicture}

%% file: n2-radius.tex
\subsection{Distribution of Radius~\texorpdfstring{$\radius_2$}{R₂}}\label{sub:n2-radius}

After having derived the joint distribution of~${(\radius_2, \phaseresult_2)}$, we can obtain the marginal density of the radius~$\radius_2$ as
\begin{align}
	\pdf_{\radius_2}(r)
	&= \int_{-\maxangle}^{\maxangle}\pdf_{\radius_2, \phaseresult_2}(r, \theta)\diff{\theta}\\
	&= \frac{2\left(\maxangle - \arccos\left(\frac{r}{2}\right)\right)}{\maxangle^2 \sqrt{4 - r^2}}\,,
\end{align}
for $2\cos\maxangle\leq r \leq 2$.
For the same range, we can then calculate the \gls{cdf} as
\begin{align}
	\cdf_{\radius_2}(r)
	&= \int_{0}^{r} \pdf_{\radius_2}(x) \diff{x}\\
	&= \frac{\left(\maxangle - \arccos\left(\frac{r}{2}\right)\right)^2}{\maxangle^2}
	\,.
\end{align}

\begin{example}[Distribution of Radius~$\radius_2$]\label{ex:radius-n2}
	\autoref{fig:cdf-radius-n2} shows the \gls{cdf} of~$\radius_2$ for different values of the maximum angle~$\maxangle$.
	For comparison, we also show the traditional case of a uniform distribution over the full circle, i.e., ${\maxangle=\pi}$, for which the \gls{cdf} has been taken from~\cite{Jammalamadaka2001}.

	Due to the assumption that the random walk is only on the right half-plane (for~${\maxangle\leq \frac{\pi}{2}}$), there is no \enquote{going back} and the resulting radius is only distributed on the interval~${[2\cos\maxangle, 2]}$.
	Thus, for small values of~$\maxangle$, e.g., ${\maxangle=0.5}$ in \autoref{fig:cdf-radius-n2}, the resulting radius is concentrated close to the maximum of~2.
	In the limiting case~${\maxangle=\frac{\pi}{2}}$, the distribution is supported over the full interval~${[0, 2]}$, just as the traditional case.
	However, in contrast to~${\maxangle=\pi}$, there is significantly less probability for small realizations of~$\radius_2$.

	\begin{figure}
		\centering
		\input{img/cdf-radius-n2.tex}
		\caption{%
			Distribution of the radius~$\radius_2$ after two steps for different values of the maximum angle~$\maxangle$.
			The case of~${\maxangle=\pi}$ corresponds to the traditional model of having a uniform distribution of the angles over the full circle.
			(\autoref{ex:radius-n2})}
		\label{fig:cdf-radius-n2}
	\end{figure}
\end{example}

%% file: img/cdf-radius-n2.tex
\begin{tikzpicture}%
	\begin{axis}[
		betterplot,
		height=.25\textheight,
		trig format plots=rad,
		xlabel={Radius~$\radius_2$},
		ylabel={\Gls{cdf} $\cdf_{\radius_2}$},
		xmin=0,
		xmax=2,
		ymin=0,
		ymax=1,
		legend pos=north west,
		]
		\addplot+[domain=1.7551651237807455:2,samples=150,mark repeat=20] {(.5-acos(x/2))^2/.5^2};
		\addlegendentry{$\maxangle=0.5$}

		\addplot+[domain=0:2,samples=150,mark repeat=20] {(pi/2-acos(x/2))^2/(pi/2)^2};
		\addlegendentry{$\maxangle=\frac{\pi}{2}$}

		\addplot+[domain=0:2,samples=150,mark repeat=20] {(2*atan(x/sqrt(4 - x^2)))/pi};
		\addlegendentry{$\maxangle=\pi$}
	\end{axis}
\end{tikzpicture}

%% file: general-n.tex
\section{Results for \texorpdfstring{$N\geq3$}{N≥3}}\label{sec:three-steps}
Even in the traditional case with independent angles, an exact calculation of the distribution is difficult for more than two steps.
While Kluyver's general solution exists, it needs to be computed numerically, e.g., through a series representation~\cite{Bennett1948}, recursive numerical integration~\cite{Simon1985}, or the Hankel transform~\cite{Besser2022ris}.
Similar to~\cite{Simon1985}, we will provide a recursive scheme to numerically calculate the distribution for moderate values of~$\numsteps$ in this section.
However, first, we give an exact characterization of the support for an arbitrary number of steps.

\subsection{Characterization of the Support}\label{sub:n3-support}
While it is ideal to have an exact characterization of the joint distribution in closed form, it can sometimes be sufficient to only have an exact characterization of the support of a distribution.
In the aforementioned application example for detecting phase misalignment, receiving a sample outside of the support indicates certain misalignment of at least one device.
Therefore, we provide a full characterization of the support of the joint distribution of radius and angle~${(\radius_{\numsteps}, \phaseresult_{\numsteps})}$ in this section.
In particular, the support consists of all points~$\Z_{\numsteps}$ that can be reached after $\numsteps$~steps according to the definition in~\eqref{eq:def-random-walk}.
In the following, we characterize the boundary of this region and analyze it in more detail.

\begin{thm}[{Support of Radius~$\radius_{\numsteps}$ and Angle~$\phaseresult_{\numsteps}$}]\label{thm:support-bound-characterization}
	Consider a two-dimensional random walk with $\numsteps$~unit-length steps and uniformly distributed phases in the interval~$[-\maxangle, \maxangle]$, $0<\maxangle\leq\frac{\pi}{2}$.
	The boundary of the support~$\support_{\numsteps}$ is given by the outer curve
	\begin{equation}\label{eq:result-support-outer}
		\supportouter[\numsteps] = \left\{\numsteps\exp\left(\imag\varphi\right) \middle| \varphi\in[-\maxangle,\maxangle]\right\}\,,
	\end{equation}
	and the inner curve
	\begin{multline}\label{eq:result-support-inner}
		\supportinner[\numsteps] = \bigcup_{k=0}^{\numsteps-1} \big\{k\exp\left(\imag\maxangle\right)+(\numsteps-1-k)\exp\left(-\imag\maxangle\right)+\exp\left(\imag\varphi\right) \big|\\ \varphi\in[-\maxangle, \maxangle]\big\}\,.
	\end{multline}
\end{thm}

\begin{proof}
First, the outer bound~$\supportouter[\numsteps]$ of the support is straightforward as it consists of all points for which the angles of all steps align, i.e.,~$\varphi=\phase_1 = \phase_2 = \cdots{} = \phase_{\numsteps}$.
In this case, the radius of the resulting points is~$\numsteps$, and it is well known that this case maximizes the radius.
Hence, the outer bound~$\supportouter[\numsteps]$ is the segment of the circle with radius~$\numsteps$ between angles~$-\maxangle$ and~$\maxangle$.

For deriving the inner bound~$\supportinner[\numsteps]$ of the support, we use a recursive approach.
For a single step~${\numsteps=1}$, the support is a one-dimensional line which is the segment of the unit circle with angles in~${[-\maxangle, \maxangle]}$, i.e., ${\supportinner[1]=\supportouter[1]}$.
The full support for~${\numsteps=2}$ is then obtained through the union of all circle segments at unit distance from each point of the support~$\supportinner[1]$.
An illustration of this construction is given in \autoref{fig:support-construction-n-1-2}.
As mentioned above, the outer bound~$\supportouter$ corresponds to the points in which the angles of the two steps align, which maximizes the radius.
In contrast, the resulting radius~$\radius_2$ is minimized for a given starting point~$\exp(\imag\phaseresult_1)$ if the second step is taken with an extreme angle, i.e., the resulting point is~$\exp(\imag\phaseresult_1)+\exp(\pm\imag\maxangle)$.
Alternatively, this inner curve can be derived by considering all possible steps from the points with extreme angles after the first step, i.e.,
\begin{equation}\label{eq:support-inner-n2-construction}
	\supportinner[2] = \left\{\exp(\pm\imag\maxangle)+\exp(\imag\varphi) \middle| \varphi\in[-\maxangle, \maxangle]\right\}\,.
\end{equation}
In \autoref{fig:support-construction-n-1-2}, this inner curve~$\supportinner[2]$ is given by the union of the yellow and red circle segments.

\begin{figure}
	\centering
	\subfloat[{Illustration of the support for ${\numsteps=1}$ and ${\numsteps=2}$.\label{fig:support-construction-n-1-2}}]{\input{img/illustration-support-construction-n1-2.tex}}
	\hfill
	\subfloat[{Illustration of the inner bound of the support~$\supportinner[\numsteps]$ for ${\numsteps=2}$ and ${\numsteps=3}$.\label{fig:support-construction-n-2-3}}]{\input{img/illustration-support-construction-n2-3.tex}}
	\vspace{-.35em}
	\caption{Illustration of constructing the inner boundary of the support~$\supportinner$.}
	\vspace{-1em}
	\label{fig:illustration-support-construction}
\end{figure}

This construction can now be extended recursively for additional steps.
For~${\numsteps=3}$, both segments of~$\supportinner[2]$ can be viewed as the (shifted) support of a single step random walk.
We therefore now apply the construction that we used to derive~$\supportinner[2]$ from~$\supportinner[1]$ to obtain~$\supportinner[3]$ from~$\supportinner[2]$.
In particular, following the scheme from~\eqref{eq:support-inner-n2-construction}, we obtain the inner boundary of the support by combining all circle segments with angles in the interval~${[-\maxangle, \maxangle]}$ around the endpoints of the respective circle segments of the previous step.
An illustration of this idea for going from two steps~${\numsteps=2}$ to three steps~${\numsteps=3}$ is shown in \autoref{fig:support-construction-n-2-3}.
The final inner boundary of the support~$\supportinner[3]$ for three steps is given as the union of the yellow, red, and green circle segments.
Explicitly, this is \vspace{-.5\abovedisplayskip}
\begin{multline*}
	\supportinner[3] = \left\{2\exp(\imag\maxangle)+\exp(\imag\varphi) \middle| \varphi\in[-\maxangle, \maxangle]\right\} \, \cup \\
	\left\{\exp(\imag\maxangle) + \exp(-\imag\maxangle) + \exp(\imag\varphi) \middle| \varphi\in[-\maxangle, \maxangle]\right\}\, \cup\\
	\left\{2\exp(-\imag\maxangle)+\exp(\imag\varphi) \middle| \varphi\in[-\maxangle, \maxangle]\right\}
	\,.
\end{multline*}
This idea can be generalized for a general number of steps~$\numsteps$ to~\eqref{eq:result-support-inner} in \autoref{thm:support-bound-characterization}.
\end{proof}

In the traditional scenario, in which each step is taken into an arbitrary direction, i.e.,~$\maxangle=\pi$, it is possible to \enquote{walk backwards}, such that the resulting radius after $\numsteps$~steps can be zero.
In contrast, this is generally not possible with the assumptions considered in this work.
Only in the limiting case of~$\maxangle=\frac{\pi}{2}$ and an even number of steps, it is possible to get back to the origin (by taking an equal number of steps in the directions of~$\pm\frac{\pi}{2}$).
Thus, for the general case, the resulting point after $\numsteps$~steps will have a non-zero minimum radius~$\radiusmin[\numsteps]$, which is given in the following corollary.

\begin{cor}[{Minimum Radius~$\radiusmin$}]
	The resulting radius~$\radius_{\numsteps}$ after $\numsteps$~steps is at least~$\radiusmin[\numsteps]$ with
	\begin{equation}\label{eq:result-radius-min}
		\radiusmin[\numsteps] = \sqrt{\ceil*{\frac{\numsteps}{2}}^2 + \floor*{\frac{\numsteps}{2}}^2 + 2\floor*{\frac{\numsteps}{2}}\ceil*{\frac{\numsteps}{2}}\cos(2\maxangle)}\,.
	\end{equation}
	If the number of steps~$\numsteps$ is even, this expression can be further simplified to
	\begin{equation}\label{eq:result-radius-min-even}
		\radiusmin[\numsteps,\textnormal{even}] = \numsteps\cos\maxangle\,.
	\end{equation}
\end{cor}

\begin{proof}
	From \autoref{thm:support-bound-characterization}, it is clear that the points with minimum radius~$\radiusmin$ lie on the inner bound of the support~$\supportinner$.
	Their radius is therefore given by
	\begin{equation*}
		\abs{(\numsteps-1-k)\exp\left(-\imag\maxangle\right)+k\exp\left(\imag\maxangle\right)+\exp\left(\imag\varphi\right)}\,,
	\end{equation*}
	for ${\varphi\in[-\maxangle, \maxangle]}$ and ${k=0,1,\dots,\numsteps-1}$.
	For an even number of steps, the minimum is achieved for an equal number of steps at angles~${+\maxangle}$ and~${-\maxangle}$.
	This corresponds to~${k=\frac{\numsteps}{2}}$ and~${\varphi=-\maxangle}$ (or equivalently ${k=\frac{\numsteps}{2}-1}$ and~${\varphi=\maxangle}$).
	The radius of the resulting point is therefore
	\begin{equation*}
		\radiusmin[\numsteps,\textnormal{even}] = \abs*{\frac{\numsteps}{2} \exp\left(-\imag\maxangle\right) + \frac{\numsteps}{2} \exp\left(\imag\maxangle\right)}\,,
	\end{equation*}
	which can be simplified to~\eqref{eq:result-radius-min-even}.
	For an odd number of steps, the minimum radius is achieved for an equal number of steps with angles~$\pm\maxangle$ plus an additional step in the direction of either~$-\maxangle$ or~$+\maxangle$.
	The radius of the resulting point is therefore
	\begin{equation*}
		\radiusmin[\numsteps,\textnormal{odd}] = \abs*{\frac{\numsteps-1}{2} \exp\left(-\imag\maxangle\right) + \left(\frac{\numsteps-1}{2}+1\right) \exp\left(\imag\maxangle\right)}\,.
	\end{equation*}
	Expanding the radius and combining it with the fact that
	\begin{equation*}
		\floor*{\frac{\numsteps}{2}} = \frac{\numsteps-1}{2} \quad\text{and}\quad \ceil*{\frac{\numsteps}{2}} = \frac{\numsteps+1}{2}
	\end{equation*}
	for odd~$\numsteps$, we obtain~\eqref{eq:result-radius-min} for general~$\numsteps$.
\end{proof}

While the description of the inner bound of the support~$\supportinner$ in~\eqref{eq:result-support-inner} is an exact representation, it can be advantageous to have a direct way to compute the set.
Therefore, we provide a simple parametrization of~$\supportinner$ in the following, which allows a straightforward computation and gives a parametrized functional connection between radius~$\radius_{\numsteps}$ and angle~$\phaseresult_{\numsteps}$.

\begin{cor}[{Parametrization of the Support}]\label{cor:support-parametrization}
	The inner boundary of the support~$\supportinner[\numsteps]$ from~\eqref{eq:result-support-inner} is given in the $\radius_{\numsteps}$-$\phaseresult_{\numsteps}$ space through the parametric function
	\begin{multline}
		\radius_{\numsteps,\textnormal{in}}(t) = \Big[\big((\numsteps-1-2k(t))\sin\maxangle + \sin\varphi(t)\big)^2 \\{+ \big((\numsteps-1)\cos\maxangle + \cos\varphi(t)\big)^2}\Big]^{\frac{1}{2}}
	\end{multline}
	\begin{equation}
		\phaseresult_{\numsteps,\textnormal{in}}(t) = \arctan\frac{(\numsteps-1-2k(t))\sin\maxangle + \sin\varphi(t)}{(\numsteps-1)\cos\maxangle + \cos\varphi(t)}
	\end{equation}
	with
	\begin{align}
		k(t) &= \ceil*{\numsteps t - 1}\\
		\varphi(t) &= \maxangle \big(2 (\numsteps t - k(t)) - 1\big)
	\end{align}
	for $0\leq t \leq 1$.
\end{cor}

\begin{example}[{Illustration of the Support}]\label{ex:support-illustration}
	\begin{figure}
		\centering
		\subfloat[{Support in Cartesian coordinates (real and imaginary parts of~$\Z$)\label{fig:illustration-support-n3}}]{\input{img/illustration-support-n3.tex}}
		\hfill
		\subfloat[{Support in polar coordinates\label{fig:illustration-support-n3-radius-angle}}]{\input{img/illustration-support-n3-radius-angle.tex}}
		\caption{%
			Support and its boundaries~$\supportinner$ and~$\supportouter$ for~${\numsteps=3}$ and~${\maxangle=0.85}$.
			Additionally, a section of the circle with the minimum radius~$\radiusmin$ is shown.
			(\autoref{ex:support-illustration})
		}
		\label{fig:illustration-support}
	\end{figure}

	A visualization of the support and its inner and outer boundary~$\supportinner$ and~$\supportouter$, respectively, is shown in \autoref{fig:illustration-support} for~${\numsteps=3}$ and~${\maxangle=0.85}$.
	First, in \autoref{fig:illustration-support-n3}, the support is shown for Cartesian coordinates, i.e., in terms of the real and imaginary components of~$\Z_3$.
	Additionally, a circle segment with the minimum radius~$\radiusmin$ is indicated.
	For the above parameters, the minimum radius is evaluated according to~\eqref{eq:result-radius-min} as~${\radiusmin[3]=2.118}$.
	Next, in \autoref{fig:illustration-support-n3-radius-angle}, the support is shown in polar coordinates, i.e., in terms of radius~$\radius_3$ and angle~$\phaseresult_3$.
	Similar to the support for two steps, cf.~\autoref{fig:pdf-joint-n2}, there is a strong dependence between radius and angle.
	In particular, for angles close to the maximum, the radius is concentrated at values close to~$\numsteps$.
\end{example}

While the parametrization of the inner bound of the support~$\supportinner$ in~\autoref{cor:support-parametrization} provides an easy way to compute~$\supportinner$ in polar coordinates, it might be desirable to describe the radius as a direct function of the angle.
However, depending on the maximum angle~$\maxangle$, this might not be possible, as will be shown in the following result.

\begin{lem}[{Condition for Unique Angles in the Support}]\label{lem:condition-support-unique}
	The radius~$\radius_{\numsteps}$ of the inner boundary of the support~$\supportinner[\numsteps]$ is only a function of the angle~$\phaseresult_{\numsteps}$ if the following condition holds
	\begin{equation}\label{eq:result-condition-support-unique}
		\maxangle \leq \frac{1}{2}\arccos\left(\frac{-1}{\numsteps-1}\right)\,.
	\end{equation}
\end{lem}
\begin{proof}
	In order to have the radius~$\radius_{\numsteps}$ of the inner boundary of the support~$\supportinner$ as a function of the angle~$\phaseresult_{\numsteps}$, each line with slope~$\tan\theta$ through the origin in the Cartesian coordinate system needs to have at most one intersection with~$\supportinner$.
	Otherwise, there exist two points in~$\supportinner$ that have the same angle~$\theta$ but different radii in polar coordinates.
	Due to symmetry around ${\Im{\Z_{\numsteps}}=0}$, we restrict the discussion to the upper half-plane in the following.

	When increasing~$\maxangle$, the first ambiguity of the radius will occur at the circle segment with the largest imaginary part, i.e., for ${k=\numsteps-1}$, cf.~\eqref{eq:result-support-inner}.
	Therefore, if there is only a single point for each angle~$\theta$ in this segment of~$\supportinner$, this will be also true for every other segment~$k$.
	According to~\eqref{eq:result-support-inner}, this critical segment is given as all points with (Cartesian) coordinates
	\begin{align*}
		x &= (\numsteps-1)\cos\maxangle + \cos\varphi\\
		y &= (\numsteps-1)\sin\maxangle + \sin\varphi
	\end{align*}
	where~$x$ and~$y$ represent the real and imaginary part, respectively, and~$\varphi\in[-\maxangle,\maxangle]$.
	As we are interested in comparing the slope of the tangent at the support with the angle~$\theta$, we express the lower half of the circle segment (${\varphi\in[-\maxangle, 0]}$) as the function
	\begin{equation*}
		y(x) = (\numsteps-1)\sin\maxangle - \sqrt{1-(x-(\numsteps-1)\cos\maxangle)^2}\,,
	\end{equation*}
	with derivative
	\begin{equation*}
		y'(x) = \frac{x-(\numsteps-1)\cos\maxangle}{\sqrt{1-(x-(\numsteps-1)\cos\maxangle)^2}}\,.
	\end{equation*}
	In order to have no ambiguity, the slope of the tangent needs to be larger than the slope of the line with angle~$\theta$, which is given as
	\begin{equation*}
		\tan\theta = \frac{y}{x} = \frac{(\numsteps-1)\sin\maxangle + \sin\varphi}{(\numsteps-1)\cos\maxangle + \cos\varphi}\,.
	\end{equation*}
	Since the slope of the tangent is increasing in the lower half of the circle segment, the critical point with the lowest slope is given at~${\varphi=-\maxangle}$.
	An illustration can be found in \autoref{fig:illustration-proof-support-unique}.
	It shows an example for which the slope of the tangent is smaller than~$\tan\theta$ (for~${\maxangle=1.05}$) and, therefore, two points exist in~$\supportinner$ with the same angle~$\theta$.
	In contrast, for the smaller value of the maximum angle~${\maxangle=0.85}$, the slope of the tangent is larger than~$\tan\theta$ and no ambiguity exists.

	\begin{figure}
		\centering
		\input{img/illustration-proof-unique-support.tex}
		\caption{Illustration of the relevant slopes of~$\supportinner$ in comparison with the angle~$\theta$, determining whether the radius of~$\supportinner$ is a function of the angle. (\autoref{lem:condition-support-unique})}
		\label{fig:illustration-proof-support-unique}
	\end{figure}

	Therefore, at this critical point, we get the two slopes as
	\begin{align*}
		y'(\numsteps\cos\maxangle) &= \frac{1}{\tan\maxangle}\\
		\tan\theta &= \frac{(\numsteps-2)\sin\maxangle}{\numsteps\cos\maxangle} = \frac{\numsteps-2}{\numsteps}\tan\maxangle\,.
	\end{align*}
	With the condition~${y'(\numsteps\cos\maxangle)\geq\tan\theta}$ and the above expressions, we obtain the condition on the maximum angle~$\maxangle$ given in~\eqref{eq:result-condition-support-unique}.
\end{proof}

\begin{cor}\label{cor:condition-support-always}
	If ${\maxangle\leq\frac{\pi}{4}}$, the radius of the inner boundary of the support~$\supportinner$ is always a function of the angle, independent of the number of steps~$\numsteps$.
\end{cor}
\begin{proof}
	The right-hand side of~\eqref{eq:result-condition-support-unique} is decreasing in the number of steps~$\numsteps$, and approaches the limit
	\begin{equation*}
		\lim_{\numsteps\to\infty}\frac{1}{2}\arccos\left(\frac{-1}{\numsteps-1}\right) = \frac{\pi}{4}\,.
	\end{equation*}
	Therefore, if the maximum angle~$\maxangle$ is smaller than this value, condition~\eqref{eq:result-condition-support-unique} in \autoref{lem:condition-support-unique} is fulfilled for all~$\numsteps$.
\end{proof}

\begin{example}[{Non-unique Radius in~$\supportinner$}]\label{ex:support-non-unique}
	An illustration of a support for which condition~\eqref{eq:result-condition-support-unique} does not hold is given in \autoref{fig:illustration-support-non-unique}.
	The number of steps is set to~${\numsteps=4}$ for which~\eqref{eq:result-condition-support-unique} gives a maximum~$\maxangle$ of around~\num{0.96}.
	However, the maximum angle in \autoref{fig:illustration-support-non-unique} is set to~${\maxangle=1.4}$.
	Therefore, there exist multiple points in the inner boundary of the support~$\supportinner$ with the same angle.
	As a consequence, the radius~$\radius_4$ of~$\supportinner[4]$ is not a function of the angle~$\phaseresult_{4}$, which can be directly seen in \autoref{fig:illustration-support-n4-radius-angle}.
	In \autoref{fig:illustration-support-n4}, this can be seen as the black dashed line has more than one intersection with~$\supportinner$.

	\begin{figure}
		\centering
		\subfloat[{Support in Cartesian coordinates (real and imaginary parts of~$\Z$)\label{fig:illustration-support-n4}}]{\input{img/illustration-support-n4.tex}}
		\hfill
		\subfloat[{Support in polar coordinates\label{fig:illustration-support-n4-radius-angle}}]{\input{img/illustration-support-n4-radius-angle.tex}}
		\caption{%
			Support and its boundaries for ${\numsteps=4}$ and~${\maxangle=1.4}$.
			Additionally, a section of the circle with the minimum radius~$\radiusmin$ is shown.
			The black dashed line indicates an angle for which more than one point in~$\supportinner$ exists, since condition~\eqref{eq:result-condition-support-unique} is violated.
			(\autoref{ex:support-non-unique})
		}
		\label{fig:illustration-support-non-unique}
	\end{figure}

	In contrast, an example for which condition~\eqref{eq:result-condition-support-unique} holds, and $\radius$~is a function of~$\phaseresult$ is given in \autoref{fig:illustration-support}.
\end{example}

\subsection{Distributions of Radius~\texorpdfstring{$\radius_{\numsteps}$}{Rₙ} and Angle~\texorpdfstring{$\phaseresult_{\numsteps}$}{θₙ}}

In the following, we provide a scheme to recursively calculate the distributions for more than two steps.
A major difference between the traditional problem with~${\maxangle=\pi}$ and the considered problem is the dependency between radius~$\radius_{\numsteps}$ and~$\phaseresult_{\numsteps}$.
Therefore, when adding a new step, i.e., when going from~$\numsteps$ to ${\numsteps+1}$~steps, we need to take their joint distribution into account.

\begin{thm}\label{thm:distributions-general-n}
	Consider a two-dimensional random walk with $\numsteps$~unit-length steps and uniformly distributed phases in the interval~${[-\maxangle, \maxangle]}$, ${0<\maxangle\leq\frac{\pi}{2}}$.
	\begin{itemize}
		\item The joint density~$\pdf_{\radius_{\numsteps}, \phaseresult_{\numsteps}}$ of the radius~$\radius_{\numsteps}$ and angle~$\phaseresult_{\numsteps}$ is calculated by~\eqref{eq:result-pdf-joint-general-n} {at the bottom of this page}.
		\capstartfalse\begin{table*}[b]
			\hrulefill
			\normalsize
			\begin{equation}\label{eq:result-pdf-joint-general-n}
				\pdf_{\radius_{\numsteps}, \phaseresult_{\numsteps}}(r, \theta) = \frac{r}{2\maxangle} \int_{-\maxangle}^{\maxangle} \frac{\pdf_{\radius_{\numsteps-1}, \phaseresult_{\numsteps-1}}\left(\sqrt{(r\cos\theta - \cos\varphi)^2 + (r\sin\theta - \sin\varphi)^2}, \arctan\frac{r\sin\theta - \sin\varphi}{r\cos\theta - \cos\varphi}\right)}{\sqrt{(r\cos\theta - \cos\varphi)^2 + (r\sin\theta - \sin\varphi)^2}} \diff{\varphi}
			\end{equation}
		\end{table*}\capstarttrue
		\item The \gls{cdf}~$\cdf_{\radius_{\numsteps}}$ of the radius~$\radius_{\numsteps}$ can be calculated by~\eqref{eq:result-cdf-radius-general-n} {at the bottom of this page}, where~$\cdf_{\phase_i}$ is the \gls{cdf} of the uniform distribution from~\eqref{eq:def-cdf-phases}.
		\capstartfalse\begin{table*}[b]
			\normalsize
			\begin{equation}\label{eq:result-cdf-radius-general-n}
				\cdf_{\radius_{\numsteps}}(r) = 1 - \iint\limits_{\support_{\numsteps-1}} \left[\cdf_{\phase_i}\left(t + \arccos\frac{r^2 - x^2 - 1}{2x}\right) - \cdf_{\phase_i}\left(t - \arccos\frac{r^2 - x^2 - 1}{2x}\right)\right] \pdf_{\radius_{\numsteps-1}, \phaseresult_{\numsteps-1}}(x, t) \diff{x}\diff{t}
			\end{equation}
		\end{table*}\capstarttrue
	\end{itemize}
	Furthermore, the distribution of the resulting angle~$\phaseresult_{\numsteps}$ can be approximated as follows.
	\begin{itemize}
		\item The \gls{cdf}~$\cdf_{\phaseresult_{\numsteps}}$ of the angle~$\phaseresult_{\numsteps}$ can be approximately calculated by
		\begin{equation}\label{eq:result-cdf-phase-general-n-approx}
			\cdf_{\phaseresult_{\numsteps}}(\theta) \approx \iint\limits_{\support_{\numsteps-1}} \cdf_{\phase_i}\big(\theta(1+x) - xt\big) \pdf_{\radius_{\numsteps-1}, \phaseresult_{\numsteps-1}}(x, t) \diff{x}\diff{t}
		\end{equation}
		The corresponding density~$\pdf_{\phaseresult_{\numsteps}}$ is calculated as
		\begin{equation}\label{eq:result-pdf-phase-general-n-approx}
			\hspace{-\leftmargin}\pdf_{\phaseresult_{\numsteps}}(\theta) \approx \!\iint\limits_{\support_{\numsteps-1}} (1+x)\pdf_{\phase_i}\big(\theta(1+x) - xt\big) \pdf_{\radius_{\numsteps-1}, \phaseresult_{\numsteps-1}}(x, t) \diff{x}\diff{t}
		\end{equation}
	\end{itemize}
\end{thm}
\begin{proof}
First, we derive the joint density~$\pdf_{\radius_{\numsteps}, \phaseresult_{\numsteps}}$ by a transformation of random variables.
Additionally, we start with Cartesian coordinates, i.e., the real and imaginary parts~$\X_{\numsteps}$ and~$\Y_{\numsteps}$, respectively.
For this, we use the following transformation~$g$
\begin{equation*}
	\begin{pmatrix}
		\X_{\numsteps}\\
		\Y_{\numsteps}\\
		\phase_{\numsteps}
	\end{pmatrix}
	= g\left(\!
	\begin{pmatrix}
		\radius_{\numsteps-1}\\
		\phaseresult_{\numsteps-1}\\
		\phase_{\numsteps}
	\end{pmatrix}\!
	\right)%
	\!=\!
	\begin{pmatrix}
		\radius_{\numsteps-1}\cos\phaseresult_{\numsteps-1} + \cos\phase_{\numsteps}\\
		\radius_{\numsteps-1}\sin\phaseresult_{\numsteps-1} + \sin\phase_{\numsteps}\\
		\phase_{\numsteps}
	\end{pmatrix}
	\!,
\end{equation*}
where we introduce~$\phase_{\numsteps}$ as an additional helper coordinate to obtain a square Jacobian matrix.
With the Jacobian determinant of~$\inv{g}$, the joint distribution of~$\transpose{(\X_{\numsteps}, \Y_{\numsteps}, \phase_{\numsteps})}$ can be obtained as
\begin{align*}
	\pdf_{\X_{\numsteps}, \Y_{\numsteps}, \phase_{\numsteps}}(x, y, \varphi)
	&= \frac{\pdf_{\radius_{\numsteps-1}, \phaseresult_{\numsteps-1}, \phase_{\numsteps}}\left(\inv{g}(x, y, \varphi)\right)}{\sqrt{(x - \cos\varphi)^2 + (y - \sin\varphi)^2}}\,.
\end{align*}
Next, we use the fact that the angle~$\phase_{\numsteps}$ is independent of the previous steps, i.e., ${\pdf_{\radius_{\numsteps-1}, \phaseresult_{\numsteps-1}, \phase_{\numsteps}}=\pdf_{\radius_{\numsteps-1}, \phaseresult_{\numsteps-1}} \pdf_{\phase_{\numsteps}}}$, and marginalize over~$\phase_{\numsteps}$ to obtain the joint \gls{pdf} of the real and imaginary components after $\numsteps$~steps.
Finally, we transform the coordinates back from the Cartesian coordinates~${(\X_{\numsteps}, \Y_{\numsteps})}$ to the polar coordinates~${(\radius_{\numsteps}, \phaseresult_{\numsteps})}$ to obtain~\eqref{eq:result-pdf-joint-general-n}.

Next, we use similar steps for calculating the marginal distribution of the radius as in the two-step case in \autoref{sub:n2-conditional-r-g-t}.
In particular, we get
\begin{align*}
	\cdf_{\radius_{\numsteps}}(r)
	&=\Pr\left(\radius_{\numsteps}\leq r\right)\\
	&= \Pr\left(\cos\left(\phase_{\numsteps}-\phaseresult_{\numsteps-1}\right) \leq \frac{r^2 - \radius_{\numsteps-1}^2 - 1}{2\radius_{\numsteps-1}}\right)\\
	\begin{split}
	&= 1 - \bigg[\cdf_{\phase_{\numsteps}}\left(\phaseresult_{\numsteps-1} + \arccos\frac{r^2 - \radius_{\numsteps-1}^2 - 1}{2\radius_{\numsteps-1}}\right)\\& \mspace{40mu plus 1fill} - \cdf_{\phase_{\numsteps}}\left(\phaseresult_{\numsteps-1} - \arccos\frac{r^2 - \radius_{\numsteps-1}^2 - 1}{2\radius_{\numsteps-1}}\right)\bigg]
	\end{split}
\end{align*}
which combined with the law of total probability yields~\eqref{eq:result-cdf-radius-general-n}.

Lastly, we derive the approximation of the \glspl{cdf} for the angle~$\phaseresult_{\numsteps}$.
We base this on the following approximation %
\begin{equation}\label{eq:approx-phase-general-n}
	\phase_{\numsteps} \approx \phaseresult_{\numsteps} - \radius_{\numsteps-1}\left(\phaseresult_{\numsteps-1} - \phaseresult_{\numsteps}\right)
	\,,
\end{equation}
which is derived using a linearization of the implicit function
\begin{equation*}
	\phaseresult_{\numsteps} - \arctan\frac{\radius_{\numsteps-1}\sin\phaseresult_{\numsteps-1} + \sin\phase_{\numsteps}}{\radius_{\numsteps-1}\cos\phaseresult_{\numsteps-1} + \cos\phase_{\numsteps}} = 0
\end{equation*}
based on the first term of the Taylor series~\cite{Koepf1994taylor}.
This provides the approximate expression of the marginal distribution
\begin{align*}
	\cdf_{\phaseresult_{\numsteps}}(t)
	&= \Pr\left(\phaseresult_{\numsteps} \leq \theta\right)\\
	&\approx \Pr\left(\frac{\phase_{\numsteps}+\radius_{\numsteps-1}\phaseresult_{\numsteps-1}}{1+\radius_{\numsteps-1}} \leq \theta\right)\\
	&= \Pr\big(\phase_{\numsteps} \leq \theta(1+\radius_{\numsteps-1}) - \radius_{\numsteps-1}\phaseresult_{\numsteps-1}\big)\,,
\end{align*}
which combined with the law of total probability yields~\eqref{eq:result-cdf-phase-general-n-approx}.
The corresponding \gls{pdf} then results from differentiating with respect to~$\theta$.
\end{proof}

\begin{example}[{Three Steps~${\numsteps=3}$}]\label{ex:n3}
	\begin{figure}
		\centering
		\subfloat[{\Gls{cdf} of the radius~$\radius_{3}$\label{fig:cdf-radius-n3}}]{\input{img/cdf-radius-n3.tex}}

		\subfloat[{\Gls{pdf} of the angle~$\phaseresult_{3}$\label{fig:pdf-angle-n3}}]{\input{img/pdf-angle-n3.tex}}

		\subfloat[{Joint \gls{pdf} of the radius and angle. Additionally, the boundary of the support from~\eqref{eq:result-support-inner} is shown.\label{fig:pdf-joint-n3}}]{\input{img/pdf-joint-n3.tex}}
		\caption{%
			Distributions of a random walk with ${\numsteps=3}$~steps and maximum angle~${\maxangle=0.5}$.
			Besides the numerical solution from \autoref{thm:distributions-general-n}, the approximation for large~$\numsteps$ from \autoref{sec:large-n}, and a histogram obtained though \gls{mc} simulations with $10^7$~samples is shown.
			(\autoref{ex:n3})
		}
		\label{fig:results-n3}
	\end{figure}

	The results of \autoref{thm:distributions-general-n} are illustrated with the example of a three-step random walk with maximum angle~${\maxangle=0.5}$, for which all distributions are shown in \autoref{fig:results-n3}.
	Besides the numerically calculated distributions from \autoref{thm:distributions-general-n}, we show a histogram of the marginals obtained through \gls{mc} simulations with $10^7$~samples for comparison.
	It can be seen that the derived result closely matches the simulated distribution.
	As an additional comparison, we plot the approximation for large~$\numsteps$, which will be derived in the following section.
	However, since we only consider three steps, this approximation shows significant deviations from the actual distribution, especially for the radius shown in \autoref{fig:cdf-radius-n3}.
	For the joint distribution in \autoref{fig:pdf-joint-n3}, we additionally show the boundary of the support as stated in \autoref{thm:support-bound-characterization}.
\end{example}

\begin{rem}
	While the marginal distributions of the angle~$\phaseresult_{\numsteps}$ can also be calculated by marginalization of the joint distribution, it has been found in our numerical experiments that the approximations~\eqref{eq:result-cdf-phase-general-n-approx} and~\eqref{eq:result-pdf-phase-general-n-approx} can be computed significantly faster.
\end{rem}

\begin{rem}[Computational Complexity]
	In general, the complexity of the recursive computation depends on the implementation of the integral.
	However, if we assume a basic quadrature implementation with $K$~points, the complexity will be exponential in the number of steps~$\numsteps$, since for each~${n=3, \dots{}, \numsteps}$, the PDF of ${n-1}$~steps needs to be computed $K$~times.
	Therefore, the total complexity in this case is~$\mathcal{O}(K^{\numsteps-2})$, which corresponds to an exponential complexity.
	While the complexity is exponential, it can still be feasible to use the recursive formulation for small~$N$ using optimized implementation and hardware.
	For large~$N$, an efficiently computable approximation is presented in the following section.
\end{rem}

%% file: img/illustration-support-construction-n1-2.tex
\begin{tikzpicture}%
	\begin{axis}[
		betterplot,
		axis equal image,
		height=.28\textheight,
		disabledatascaling,
		xlabel={$\radius\cos\phaseresult$},
		ylabel={$\radius\sin\phaseresult$},
		ylabel shift=-.9em,
		xmin=0,
		xmax=2.1,
		ymin=-2.1,
		ymax=2.1,
		]
		\pgfmathsetmacro{\MAXANGLE}{50}
		\addplot[dashed,domain=0:3,samples=2,opacity=.25] {tan(\MAXANGLE)*x};
		\addplot[dashed,domain=0:3,samples=2,opacity=.25] {-tan(\MAXANGLE)*x};

		\draw[black,thick] (0,0) ++(\MAXANGLE:1) arc (\MAXANGLE:-\MAXANGLE:1);

		\draw[plot1,thin,densely dashed] (\MAXANGLE:1) -- ++ (\MAXANGLE:1);
		\draw[plot1,thin,densely dashed] (\MAXANGLE:1) -- ++ (-\MAXANGLE:1);
		\draw[plot1,very thick,line join=round] (\MAXANGLE:1) ++(\MAXANGLE:1) arc (\MAXANGLE:-\MAXANGLE:1);

		\draw[plot2,thin,densely dashed] (-\MAXANGLE:1) -- ++ (\MAXANGLE:1);
		\draw[plot2,thin,densely dashed] (-\MAXANGLE:1) -- ++ (-\MAXANGLE:1);
		\draw[plot2,very thick,line join=round] (-\MAXANGLE:1) ++(\MAXANGLE:1) arc (\MAXANGLE:-\MAXANGLE:1);

		\draw[plot5,densely dashed] (0:1) -- ++ (\MAXANGLE:1);
		\draw[plot5,thin,densely dashed] (0:1) -- ++ (-\MAXANGLE:1);
		\draw[plot5,thick,line join=round] (0:1) ++(\MAXANGLE:1) arc (\MAXANGLE:-\MAXANGLE:1);

		\draw[plot0,thin,densely dashed] (\MAXANGLE/2:1) -- ++ (\MAXANGLE:1);
		\draw[plot0,thin,densely dashed] (\MAXANGLE/2:1) -- ++ (-\MAXANGLE:1);
		\draw[plot0,thick,line join=round] (\MAXANGLE/2:1) ++(\MAXANGLE:1) arc (\MAXANGLE:-\MAXANGLE:1);

		\node[inner sep=0pt,anchor=south] (labelInner1) at (.44, 1.02) {$\supportinner[1]$};
		\draw (labelInner1) to[out=-90,in=180] (.8,.5);

		\node[inner sep=0pt] (labelInner21) at (.95, 1.8) {$\supportinner[2]$};
		\draw (labelInner21) -- (1.25,1.55);
		\node[inner sep=0pt] (labelInner22) at (.95, -1.8) {$\supportinner[2]$};
		\draw (labelInner22) -- (1.25,-1.55);
	\end{axis}
\end{tikzpicture}

%% file: img/illustration-support-construction-n2-3.tex
\begin{tikzpicture}%
	\begin{axis}[
		betterplot,
		axis equal image,
		height=.28\textheight,
		disabledatascaling,
		xlabel={$\radius\cos\phaseresult$},
		ylabel={$\radius\sin\phaseresult$},
		ylabel shift=-.9em,
		xmin=.5,
		xmax=3,
		ymin=-2.5,
		ymax=2.5,
		]
		\pgfmathsetmacro{\MAXANGLE}{50}
		\addplot[dashed,domain=0:3,samples=2,opacity=.25] {tan(\MAXANGLE)*x};
		\addplot[dashed,domain=0:3,samples=2,opacity=.25] {-tan(\MAXANGLE)*x};

		\draw[plot0,very thick,name path=inner2,line join=round] (\MAXANGLE:1) ++(\MAXANGLE:1) arc (\MAXANGLE:-\MAXANGLE:1)
		-- ++(0,0) arc (\MAXANGLE:-\MAXANGLE:1)
		;

		\draw[plot1,thin,densely dashed] (\MAXANGLE:2) -- ++ (\MAXANGLE:1);
		\draw[plot1,thin,densely dashed] (\MAXANGLE:2) -- ++ (-\MAXANGLE:1);
		\draw[plot1,thick,line join=round] (\MAXANGLE:2) ++(\MAXANGLE:1) arc (\MAXANGLE:-\MAXANGLE:1);

		\draw[plot2,densely dashed] (0:1.286) -- ++ (\MAXANGLE:1);
		\draw[plot2,thin,densely dashed] (0:1.286) -- ++ (-\MAXANGLE:1);
		\draw[plot2,thick,line join=round] (0:1.286) ++(\MAXANGLE:1) arc (\MAXANGLE:-\MAXANGLE:1);

		\draw[plot5,thin,densely dashed] (-\MAXANGLE:2) -- ++ (\MAXANGLE:1);
		\draw[plot5,thin,densely dashed] (-\MAXANGLE:2) -- ++ (-\MAXANGLE:1);
		\draw[plot5,thick,line join=round] (-\MAXANGLE:2) ++(\MAXANGLE:1) arc (\MAXANGLE:-\MAXANGLE:1);

		\node[inner sep=0pt] (labelInner2) at (.94, .3) {$\supportinner[2]$};
		\draw (labelInner2) -- (1.6,.75);
		\draw (labelInner2) -- (1.6,-.75);

		\node[inner sep=0pt] (labelInner3) at (2.7,.75) {$\supportinner[3]$};
		\draw (labelInner3) -- (2.12,.6);
		\draw (labelInner3) -- (2.17, 1);
		\draw (labelInner3) -- (2.17, -1);
	\end{axis}
\end{tikzpicture}

%% file: img/illustration-support-n3.tex
\begin{tikzpicture}%
	\begin{axis}[
		betterplot,
		axis equal image,
		height=.285\textheight,
		disabledatascaling,
		xlabel={$\Re\Z$},
		ylabel={$\Im\Z$},
		ylabel shift=-.9em,
		xmin=0,
		xmax=3,
		ymin=-3,
		ymax=3,
		]
		\pgfmathsetmacro{\MAXANGLE}{48.7} %
		\addplot[dashed,domain=0:3,samples=2,opacity=.25] {tan(\MAXANGLE)*x};
		\addplot[dashed,domain=0:3,samples=2,opacity=.25] {-tan(\MAXANGLE)*x};

		\draw[plot0,very thick,name path=inner,line join=round] (\MAXANGLE:2) ++(\MAXANGLE:1) arc (\MAXANGLE:-\MAXANGLE:1)
		-- ++(0,0) arc (\MAXANGLE:-\MAXANGLE:1)
		-- ++(0,0) arc (\MAXANGLE:-\MAXANGLE:1)
		;
		\draw[plot1,very thick,name path=outer] (0,0) ++(\MAXANGLE:3) arc (\MAXANGLE:-\MAXANGLE:3);

		\tikzfillbetween[of=inner and outer]{plot2, opacity=0.2};

		\draw[plot0,dashed] (0,0) ++(\MAXANGLE:{sqrt(5+4*cos(2*\MAXANGLE))}) arc (\MAXANGLE:-\MAXANGLE:{sqrt(5+4*cos(2*\MAXANGLE))});

		\node[inner sep=0pt] (labelOuter) at (2.6,2.4) {$\supportouter$};
		\draw (labelOuter) -- (45:3.05);

		\node[inner sep=0pt] (labelInner) at (1.8,1.7) {$\supportinner$};
		\draw (labelInner) -- (2.2,1.75);

		\draw[|<->|] (0, 0) -- node[fill=white,inner sep=2pt] {$\radiusmin$} (-40:{sqrt(5+4*cos(2*\MAXANGLE))});

		\node[plot2] at (2.6,0) {$\support$};
	\end{axis}
\end{tikzpicture}

%% file: img/illustration-support-n3-radius-angle.tex
\begin{tikzpicture}%
	\begin{axis}[
		betterplot,
		width=.55\linewidth,
		height=.28\textheight,
		xlabel={$\phaseresult_{3}$},
		ylabel={$\radius_3$},
		ylabel shift=-1.5em,
		xmin=-.85,
		xmax=.85,
		ymin=2,
		ymax=3,
		enlarge y limits=.05,
		extra y ticks={2.118},
		extra y tick labels={$\radiusmin[3]$},
		]
		\pgfmathsetmacro{\MAXANGLE}{48.7} %

		\pgfplotstableread{data/results-support-a0.850-n3.dat}\ta

		\addplot+[mark repeat=30,name path=inner] table[x=angleInner,y=radInner] {\ta};
		\addplot+[domain=-.85:.85,samples=15,name path=outer] {3};

		\tikzfillbetween[of=inner and outer]{plot2, opacity=0.2};
		\node[plot2] at (0, 2.7) {$\support_{3}$};
	\end{axis}
\end{tikzpicture}

%% file: img/illustration-proof-unique-support.tex
\begin{tikzpicture}[spy using outlines={circle, magnification=2.5, connect spies}]%
	\begin{axis}[
		betterplot,
		axis equal image,
		height=.27\textheight,
		disabledatascaling,
		xlabel={$x$},
		ylabel={$y$},
		xmin=0,
		xmax=3.1,
		ymin=0,
		ymax=3.3,
		legend pos=north west,
		]
		\pgfmathsetmacro{\MAXANGLE}{48.7} %
		\pgfmathsetmacro{\LARGERANGLE}{60} %

		\draw[plot0,very thick,name path=inner,line join=round] (\MAXANGLE:3) ++(\MAXANGLE:1) arc (\MAXANGLE:-\MAXANGLE:1)
		-- ++(0,0) arc (\MAXANGLE:-\MAXANGLE:1)
		;
		\addlegendimage{very thick,plot0}
		\addlegendentry{$\maxangle=0.85$}

		\draw[plot1,very thick,name path=inner,line join=round] (\LARGERANGLE:3) ++(\LARGERANGLE:1) arc (\LARGERANGLE:-\LARGERANGLE:1)
		-- ++(0,0) arc (\LARGERANGLE:-\LARGERANGLE:1)
		;
		\addlegendimage{very thick,plot1}
		\addlegendentry{$\maxangle=1.05$}

		\addplot[densely dashdotted,domain=0:4,samples=2] {((2*sin(\MAXANGLE))/(4*cos(\MAXANGLE)))*x};
		\draw (0:.9) arc (0:{atan((2*sin(\MAXANGLE))/(4*cos(\MAXANGLE)))}:.9);
		\node[inner sep=0pt,anchor=south] at (.65,.1) {$\theta$};

		\addplot[plot0,thick,densely dashed,domain=2:4,samples=2] {1/tan(\MAXANGLE)*(x-4*cos(\MAXANGLE)) + 2*sin(\MAXANGLE)};

		\addplot[domain=0:4,samples=2,densely dashdotted] {((2*sin(\LARGERANGLE))/(4*cos(\LARGERANGLE)))*x};
		\addplot[plot1,thick,densely dashed,domain=1.5:2.75,samples=2] {1/tan(\LARGERANGLE)*(x-4*cos(\LARGERANGLE)) + 2*sin(\LARGERANGLE)};

		\node[semithick, circle, draw, minimum size=.667cm, inner sep=0pt] (spypoint) at (2.15,1.85) {};
		\node[semithick, circle, draw, minimum size=2cm, inner sep=0pt] (spyviewer) at (.8,1.8) {};
		\draw (spypoint) edge (spyviewer);

		\pgfmathsetmacro{\spyfactor}{1.7320508075688772}
		\begin{scope}
			\clip (.8,1.8) circle (1cm-.5\pgflinewidth);
			\pgfmathparse{\spyfactor^2/(\spyfactor-1)}
			\begin{scope}[scale around={\spyfactor:($(.8,1.8)!\spyfactor^2/(\spyfactor^2-1)!(2.15,1.85)$)}]
				\addplot[plot1,domain=2:2.4,thick] {3*sin(\LARGERANGLE)-sqrt(1-(x-3*cos(\LARGERANGLE))^2)};
				\addplot[plot1,domain=2:2.4,thick] {sin(\LARGERANGLE)+sqrt(1-(x-3*cos(\LARGERANGLE))^2)};
				\addplot[domain=0:4,samples=2,densely dashdotted] {((2*sin(\LARGERANGLE))/(4*cos(\LARGERANGLE)))*x};
				\addplot[plot1,thick,densely dashed,domain=1.5:2.75,samples=2] {1/tan(\LARGERANGLE)*(x-4*cos(\LARGERANGLE)) + 2*sin(\LARGERANGLE)};
			\end{scope}
		\end{scope}
	\end{axis}
\end{tikzpicture}

%% file: img/illustration-support-n4.tex
\begin{tikzpicture}%
	\begin{axis}[
		betterplot,
		axis equal image,
		height=.27\textheight,
		disabledatascaling,
		xlabel={$\radius_4\cos\phaseresult_{4}$},
		ylabel={$\radius_4\sin\phaseresult_{4}$},
		ylabel shift=-.9em,
		xmin=0,
		xmax=4,
		ymin=-4,
		ymax=4,
		]
		\pgfmathsetmacro{\MAXANGLE}{80.21}
		\addplot[dashed,domain=0:4,samples=2,opacity=.25] {tan(\MAXANGLE)*x};
		\addplot[dashed,domain=0:4,samples=2,opacity=.25] {-tan(\MAXANGLE)*x};

		\draw[plot0,very thick,name path=inner,line join=round] (\MAXANGLE:3) ++(\MAXANGLE:1) arc (\MAXANGLE:-\MAXANGLE:1)
		-- ++(0,0) arc (\MAXANGLE:-\MAXANGLE:1)
		-- ++(0,0) arc (\MAXANGLE:-\MAXANGLE:1)
		-- ++(0,0) arc (\MAXANGLE:-\MAXANGLE:1)
		;
		\draw[plot1,very thick,name path=outer] (0,0) ++(\MAXANGLE:4) arc (\MAXANGLE:-\MAXANGLE:4);

		\tikzfillbetween[of=inner and outer]{plot2, opacity=0.2};

		\draw[plot0,dashed] (0,0) ++(\MAXANGLE:{4*cos(\MAXANGLE)}) arc (\MAXANGLE:-\MAXANGLE:{4*cos(\MAXANGLE)});

		\addplot[densely dashed,domain=0:4,samples=2] {tan(63)*x}; %
	\end{axis}
\end{tikzpicture}

%% file: img/illustration-support-n4-radius-angle.tex
\begin{tikzpicture}%
	\begin{axis}[
		betterplot,
		width=.55\linewidth,
		height=.27\textheight,
		xlabel={$\phaseresult_{4}$},
		ylabel={$\radius_4$},
		ylabel shift=-1.5em,
		xmin=-1.4,
		xmax=1.4,
		ymin=.5,
		ymax=4,
		enlarge y limits=.05,
		extra y ticks={.680},
		extra y tick labels={$\radiusmin[4]$},
		]

		\pgfplotstableread{data/results-support-a1.400-n4.dat}\ta
		\addplot+[mark repeat=30,name path=inner] table[x=angleInner,y=radInner] {\ta};
		\addplot+[domain=-1.4:1.4,samples=15,name path=outer] {4};

		\tikzfillbetween[of=inner and outer]{plot2, opacity=0.2};

		\draw[dashed] (1.1,0) -- (1.1,4.2); %
	\end{axis}
\end{tikzpicture}

%% file: img/cdf-radius-n3.tex
\begin{tikzpicture}%
	\begin{axis}[
		betterplot,
		height=.22\textheight,
		trig format plots=rad,
		xlabel={Radius~$\radius_3$},
		ylabel={\Gls{cdf} $\cdf_{\radius_3}$},
		xmin=2.65,
		xmax=3,
		ymin=0,
		ymax=1,
		legend pos=north west,
		]
		\pgfplotstableread{data/results-n3-a0.500.dat}\tnum
		\pgfplotstableread{data/results-mc-n3-a0.500.dat}\tmc

		\addplot+[mark repeat=30] table[x=radius,y=cdfRadius] {\tnum};
		\addlegendentry{Numeric.}

		\addplot+[mark repeat=30] table[x=radius,y=cdfRadiusApprox] {\tnum};
		\addlegendentry{Approx.}

		\addplot+[mark repeat=3,const plot mark mid] table[x=radius,y=histRad] {\tmc};
		\addlegendentry{MC}
	\end{axis}
\end{tikzpicture}

%% file: img/pdf-angle-n3.tex
\begin{tikzpicture}%
	\begin{axis}[
		betterplot,
		height=.225\textheight,
		xlabel={Angle~$\phaseresult_3$},
		ylabel={\Gls{pdf} $\pdf_{\phaseresult_3}$},
		xmin=-.5,
		xmax=.5,
		ymin=0,
		ymax=2.5,
		]
		\pgfplotstableread{data/results-n3-a0.500.dat}\tnum
		\pgfplotstableread{data/results-mc-n3-a0.500.dat}\tmc

		\addplot+[mark repeat=30] table[x=angle,y=pdfAngle] {\tnum};
		\addlegendentry{Numeric.}

		\addplot+[mark repeat=30] table[x=angle,y=pdfAngleApprox] {\tnum};
		\addlegendentry{Approx.}

		\addplot+[mark repeat=3,const plot mark right] table[x=angle,y=histAngle] {\tmc};
		\addlegendentry{MC}
	\end{axis}
\end{tikzpicture}

%% file: img/pdf-joint-n3.tex
\begin{tikzpicture}
	\begin{axis}[
		width=.8\linewidth,
		height=.225\textheight,
		trig format plots=rad,
		view={0}{90},
		xmin=-.5,
		xmax=.5,
		ymax=3,
		colormap name=viridis,
		colorbar style={xtick=data},
		ylabel={Radius~$\radius_3$},
		xlabel={Angle~$\phaseresult_3$},
		axis on top,
		]
		\pgfplotstableread{data/results-joint-n3-a0.500.dat}\tbl
		\pgfplotstableread{data/results-support-n3-a0.500.dat}\tsupp

		\addplot3[surf,faceted color=none] table [x=angle,y=radius,z=pdfJoint]{\tbl};
		\addplot[plot1,thick] table[x=angle,y=radius] {\tsupp};
	\end{axis}
\end{tikzpicture}

%% file: large-n.tex
\section{Approximation for Large \texorpdfstring{$\numsteps$}{N}}\label{sec:large-n}

\begin{thm}\label{thm:distributions-large-n}
	Consider a two-dimensional random walk with $\numsteps$~unit-length steps and uniformly distributed phases in the interval~${[-\maxangle, \maxangle]}$, ${0<\maxangle\leq\frac{\pi}{2}}$.
	The distributions of the resulting radius~$\radius_{\numsteps}$ and resulting angle~$\phaseresult_{\numsteps}$ can be approximated for large~$\numsteps$ as follows.
	\begin{itemize}
		\item The \gls{cdf}~$\cdf_{\radius_{\numsteps}}$ of the radius~$\radius_{\numsteps}$ is approximated by
		\begin{equation}\label{eq:result-cdf-radius-large-n}
			\cdf_{\radius_{\numsteps}}(r) = \cdf_{\radius^2_{\numsteps}}(r^2) \sim \gxsdist(\vec{w}, \vec{k}, \vec{\lambda}, s, m)
		\end{equation}
		where $\radius^2_{\numsteps}$ is distributed according to a generalized chi-square distribution~${\gxsdist(\vec{w}, \vec{k}, \vec{\lambda}, s, m)}$~\cite{Das2025} with parameters~${\vec{w}=(\numsteps\var_x, \numsteps\var_y)}$, ${\vec{k}=(1, 1)}$, ${\vec{\lambda}=(\frac{\numsteps\mean_x^2}{\var_x}, 0)}$, ${s=0}$, and ${m=0}$.
		Furthermore, the parameters~$\mean_x$, $\var_x$, and $\var_y$ are given as \vspace{-1.5ex}
		\begin{align}
			\mean_x &= \frac{\sin\maxangle}{\maxangle} \label{eq:result-mean-x}\\
			\var_x &= \frac{\maxangle+\cos\maxangle \sin\maxangle}{2\maxangle} - \left(\frac{\sin\maxangle}{\maxangle}\right)^2 \label{eq:result-var-x}\\
			\var_y &= \frac{\maxangle-\cos\maxangle \sin\maxangle}{2\maxangle} \label{eq:result-var-y}\,.
		\end{align}
		The corresponding density~$\pdf_{\radius_{\numsteps}}$ is given as
		\begin{equation}\label{eq:result-pdf-radius-large-n}
			\pdf_{\radius_{\numsteps}}(r) = 2 r \pdf_{\radius^2_{\numsteps}}(r^2)\,.
		\end{equation}
		\item The \gls{cdf}~$\cdf_{\phaseresult_{\numsteps}}$ of the angle~$\phaseresult_{\numsteps}$ is approximated by
		\begin{equation}\label{eq:result-cdf-phase-large-n}
			\cdf_{\phaseresult_{\numsteps}}(\theta) = \cdfnormal\left(\frac{\numsteps\mean_x \tan \theta}{\sqrt{\numsteps\left(\tan^2(\theta) \var_x+ \var_y\right)}}\right)\,,
		\end{equation}
		where $\cdfnormal$~represents the \gls{cdf} of the standard normal distribution and the parameters~$\mean_x$, $\var_x$, and $\var_y$ are given in~\eqref{eq:result-mean-x}, \eqref{eq:result-var-x}, and~\eqref{eq:result-var-y}, respectively.
		The corresponding density~$\pdf_{\phaseresult_{\numsteps}}$ is given in~\eqref{eq:result-pdf-phase-large-n} {at the top of this page},
		\capstartfalse\begin{table*}%
			\normalsize
			\begin{equation}\label{eq:result-pdf-phase-large-n}
				\pdf_{\phaseresult_{\numsteps}}(\theta) = \pdfnormal\left(\frac{\numsteps\mean_x \tan \theta}{\sqrt{\numsteps\left(\tan^2(\theta) \var_x+ \var_y\right)}}\right) \frac{\sqrt{\numsteps}\mean_x \var_y}{\left(\tan^2(\theta) \var_x+ \var_y\right)^{\frac{3}{2}}\cos^2(\theta)}%
			\end{equation}
			\hrulefill
		\end{table*}\capstarttrue
		in which $\pdfnormal$~represents the \gls{pdf} of the standard normal distribution.
		\item The joint density~$\pdf_{\radius_{\numsteps}, \phaseresult_{\numsteps}}$ of the radius~$\radius_{\numsteps}$ and angle~$\phaseresult_{\numsteps}$ is approximated by
		\begin{equation}\label{eq:result-pdf-joint-large-n}\hspace*{-\leftmargin}
			\pdf_{\radius_{\numsteps}, \phaseresult_{\numsteps}}(r, \theta)%
			= \frac{r}{\numsteps^2\std_x\std_y}\pdfnormal\left(\frac{r\cos\theta - \numsteps\mean_x}{\numsteps\std_x}\right) \pdfnormal\left(\frac{r\sin \theta}{\numsteps\std_y}\right).
		\end{equation}
	\end{itemize}
\end{thm}

\begin{proof}
	In the remainder of this section, we prove \autoref{thm:distributions-large-n}.
	The basis of the proof is the application of the \gls{clt} to the real and imaginary part of the result after $\numsteps$~steps.
	With this, we are then able to approximate the distributions of~$\radius_{\numsteps}$ and~$\phaseresult_{\numsteps}$, as well as their joint distribution.
\end{proof}

\input{large-n-clt.tex}
\input{large-n-radius.tex}
\input{large-n-angle.tex}
\input{large-n-joint.tex}

%% file: large-n-clt.tex
\subsection{Application of the Central Limit Theorem}\label{sub:large-n-clt}

As a first step of the proof, we apply the \gls{clt} to the real and imaginary components of the complex number, which represents the result~$\Z_{\numsteps}$ of the random walk after $\numsteps$~steps.
They are obtained through summation of the real and imaginary components of the individual steps~$i$, which are given as
\begin{equation}\label{eq:def-x-y-components}
	\X_i = \cos\phase_i \quad\text{and}\quad \Y_i = \sin\phase_i\,,
\end{equation}
respectively.
Furthermore, their sample averages are
\begin{equation}\label{eq:def-x-y-sample-averages}
	\avg{\X} = \frac{1}{\numsteps} \sum_{i=1}^{\numsteps} \X_i \quad\text{and}\quad \avg{\Y} = \frac{1}{\numsteps} \sum_{i=1}^{\numsteps} \Y_i\,.
\end{equation}
By the assumptions on the random walk, all~$\X_i$ are \gls{iid}, as each individual step is taken independently and with the same distribution as the others.
Analogously, the same holds for~$\Y_i$, even if~$\X_i$ and~$\Y_i$ are not independent.
Therefore, applying the \gls{clt} to the random vector~${(\X_i, \Y_i)}$ yields
\begin{equation}\label{eq:clt-x-y-components}
	\sqrt{\numsteps}\left(\begin{pmatrix}\avg{\X}\\\avg{\Y}\end{pmatrix}-\begin{pmatrix}\mean_x\\\mean_y\end{pmatrix}\right) \xrightarrow{d} \normaldist\left(0, \Sigma\right)
\end{equation}
with expected values~${\mean_x=\expect{\X_i}}$ and~${\mean_y=\expect{\Y_i}}$, variances~${\var_x=\expect{(\X_i-\mean_x)^2}}$, ${\var_y=\expect{(\Y_i-\mean_y)^2}}$, and covariance matrix
\begin{equation*}
	\Sigma = \begin{pmatrix}
		\var_x & \cov(\X, \Y)\\
		\cov(\Y, \X) & \var_y
	\end{pmatrix}\,.
\end{equation*}

For the uniform distribution of~$\phase_i$ in~\eqref{eq:def-pdf-phases}, we obtain the following values for the mean values
\begin{align}
	\mean_x = \expect{\X_i} &= \int_{-\maxangle}^{\maxangle}\frac{1}{2\maxangle}\cos t \diff{t} = \frac{\sin\maxangle}{\maxangle}\\
	\mean_y = \expect{\Y_i} &= \int_{-\maxangle}^{\maxangle}\frac{1}{2\maxangle}\sin t \diff{t} = 0\,,
\end{align}
for the variances
\begin{align}
	\var_x = \expect{(\X_i-\mean_x)^2}
	&= \frac{\maxangle+\cos\maxangle \sin\maxangle}{2\maxangle} - \left(\frac{\sin\maxangle}{\maxangle}\right)^2\\
	\var_y = \expect{(\Y_i-\mean_y)^2}
	&= \frac{\maxangle-\cos\maxangle \sin\maxangle}{2\maxangle}\,,
\end{align}
and the covariance
\begin{equation}
	\cov(\X, \Y) = \frac{1}{2\maxangle}\int_{-\maxangle}^{\maxangle}\left(\cos t - \frac{\sin\maxangle}{\maxangle}\right) \sin t \diff{t} = 0\,.
\end{equation}
This result shows that the sample averages of the real and imaginary components become jointly normal and independent as~$\numsteps$ grows.

%% file: large-n-radius.tex
\subsection{Approximate Distribution of Radius~\texorpdfstring{$\radius_{\numsteps}$}{Rₙ}}

First, we use the application of the \gls{clt} to approximate the distribution of the radius~$\radius_{\numsteps}$ after a large number of steps, i.e.,~${\numsteps\gg 1}$.

For this, we express the radius in terms of~$\avg{\X}$ and~$\avg{\Y}$ as
\begin{align}
	\radius_{\numsteps}^2
	&= (\numsteps\avg{\X})^2 + (\numsteps\avg{\Y})^2 \label{eq:radius-large-n-avg-xy}\\
	&= \numsteps\var_x\left(\rv{V}+\frac{\sqrt{\numsteps}\mean_x}{\std_x}\right)^2 +  \numsteps\var_y\tilde{\rv{V}}^2\label{eq:radius-large-n-sum-ncx2}\,,
\end{align}
where $\rv{V}$~and~$\tilde{\rv{V}}$ are independent standard normally distributed random variables.
Based on~\eqref{eq:radius-large-n-sum-ncx2}, $\radius_{\numsteps}^2$ is given as the weighted sum of two non-central chi-square distributed random variables with one degree of freedom and non-centrality parameters~$\frac{\numsteps\mean_x^2}{\var_x}$ and~$0$, respectively.
Thus, $\radius_{\numsteps}^2$ is distributed according to a generalized chi-square distribution~${\gxsdist(\vec{w}, \vec{k}, \vec{\lambda}, s, m)}$~\cite{Das2025} with parameters~${\vec{w}=(\numsteps\var_x, \numsteps\var_y)}$, ${\vec{k}=(1, 1)}$, ${\vec{\lambda}=(\frac{\numsteps\mean_x^2}{\var_x}, 0)}$, ${s=0}$, and ${m=0}$.
While this distribution does not admit a closed-form expression, it can be efficiently computed numerically~\cite{Imhof1961,Ruben1962,Das2025}.
A Python implementation is made freely available at~\cite{BesserGithub}.

To obtain the distribution of~$\radius_{\numsteps}$ from~$\radius_{\numsteps}^2$, we simply apply a transformation of variables, which yields
\begin{align*}
	\cdf_{\radius_{\numsteps}}(r) &= \cdf_{\radius_{\numsteps}^2}(r^2)\\
	\pdf_{\radius_{\numsteps}}(r) &= 2r\pdf_{\radius_{\numsteps}^2}(r^2)
\end{align*}
for the \gls{cdf} and \gls{pdf}, respectively.

\begin{example}[{Approximate Distribution of Radius~$\radius_{\numsteps}$}]\label{ex:radius-n-large}
	As a numerical illustration, we plot the approximate \gls{cdf} of the radius~$\radius_{\numsteps}$ from~\eqref{eq:result-cdf-radius-large-n} together with a numerical \gls{cdf} obtained through \gls{mc} simulations with $10^7$~samples in \autoref{fig:cdf-radius-large-n}.
	We show the results for a small number of steps, namely~${\numsteps=5}$, in \autoref{fig:cdf-radius-n5}.
	As expected, the approximation shows significant deviations from the numerical results, as the number of steps is too small.
	In contrast, for the larger~${\numsteps=30}$ in \autoref{fig:cdf-radius-n30}, the proposed approximation for large~$\numsteps$ is valid.

	\begin{figure}
		\centering
		\subfloat[{$\numsteps=5$\label{fig:cdf-radius-n5}}]{\input{img/cdf-radius-n5.tex}}

		\vspace{-.75em}
		\subfloat[{$\numsteps=30$\label{fig:cdf-radius-n30}}]{\input{img/cdf-radius-n30.tex}}
		\caption{%
			Approximation of the distribution of the radius~$\radius_{\numsteps}$ for large~$\numsteps$ from~\eqref{eq:result-cdf-radius-large-n} together with a histogram obtained through \gls{mc} simulations with $10^7$~samples.
			(\autoref{ex:radius-n-large})
		}
		\label{fig:cdf-radius-large-n}
	\end{figure}
\end{example}

%% file: img/cdf-radius-n5.tex
\begin{tikzpicture}%
	\begin{axis}[
		betterplot,
		height=.225\textheight,
		trig format plots=rad,
		xlabel={Radius~$\radius_5$},
		ylabel={\Gls{cdf} $\cdf_{\radius_5}$},
		xmin=2,
		xmax=5,
		ymin=0,
		ymax=1,
		legend pos=north west,
		]
		\pgfplotstableread[x=time]{data/results-approx-n-large-a0.500-n5.dat}\ta
		\pgfplotstableread[x=time]{data/results-mc-n-large-a0.500-n5.dat}\tam

		\pgfplotstableread[x=time]{data/results-approx-n-large-a1.250-n5.dat}\tb
		\pgfplotstableread[x=time]{data/results-mc-n-large-a1.250-n5.dat}\tbm

		\addplot+[mark repeat=30] table[x=radius,y=cdfRad] {\ta};
		\addlegendentry{Approx.\,--\,$\maxangle=0.5$}

		\addplot+[mark repeat=3,const plot mark mid] table[x=radius,y=histRad] {\tam};
		\addlegendentry{MC\,--\,$\maxangle=0.5$}

		\addplot+[mark repeat=30] table[x=radius,y=cdfRad] {\tb};
		\addlegendentry{Approx.\,--\,$\maxangle=1.25$}

		\addplot+[mark repeat=3,const plot mark mid] table[x=radius,y=histRad] {\tbm};
		\addlegendentry{MC\,--\,$\maxangle=1.25$}

	\end{axis}

\end{tikzpicture}

%% file: img/cdf-radius-n30.tex
\begin{tikzpicture}%
	\begin{axis}[
		betterplot,
		height=.225\textheight,
		trig format plots=rad,
		xlabel={Radius~$\radius_{30}$},
		ylabel={\Gls{cdf} $\cdf_{\radius_{30}}$},
		xmin=20,
		xmax=30,
		ymin=0,
		ymax=1,
		legend style={
			anchor=south,
			at={(.57,.04)},
			font=\small,
		},
		]
		\pgfplotstableread[x=time]{data/results-approx-n-large-a0.500-n30.dat}\ta
		\pgfplotstableread[x=time]{data/results-mc-n-large-a0.500-n30.dat}\tam
		\pgfplotstableread[x=time]{data/results-approx-n-large-a1.250-n30.dat}\tb
		\pgfplotstableread[x=time]{data/results-mc-n-large-a1.250-n30.dat}\tbm

		\addplot+[mark repeat=30] table[x=radius,y=cdfRad] {\ta};
		\addlegendentry{Approx.\,--\,$\maxangle=1.25$}

		\addplot+[mark repeat=3,const plot mark mid] table[x=radius,y=histRad] {\tam};
		\addlegendentry{MC\,--\,$\maxangle=0.5$}

		\addplot+[mark repeat=30] table[x=radius,y=cdfRad] {\tb};
		\addlegendentry{Approx.\,--\,$\maxangle=1.25$}

		\addplot+[mark repeat=3,const plot mark mid] table[x=radius,y=histRad] {\tbm};
		\addlegendentry{MC\,--\,$\maxangle=1.25$}
	\end{axis}
\end{tikzpicture}

%% file: large-n-angle.tex
\subsection{Approximate Distribution of Angle~\texorpdfstring{$\phaseresult_{\numsteps}$}{θₙ}}

After deriving an approximate distribution of the radius~$\radius_{\numsteps}$ after a large number of steps, we now focus on the resulting angle~$\phaseresult_{\numsteps}$.
Based on the real and imaginary parts, it is given as
\begin{equation}\label{eq:angle-large-n-avg-xy}
	\phaseresult_{\numsteps} = \arctan\frac{\sum_{i=1}^{\numsteps}\Y_i}{\sum_{i=1}^{\numsteps}\X_i} = \arctan\frac{\avg{\Y}}{\strut\avg{\X}}\,, %
\end{equation}
where we use~$\X_i$,~$\Y_i$, and their sample averages defined in~\eqref{eq:def-x-y-components} and~\eqref{eq:def-x-y-sample-averages}, respectively.
According to~\eqref{eq:clt-x-y-components}, the sample averages~${(\avg{\X}, \avg{\Y})}$ are approximately jointly Gaussian distributed for large~$\numsteps$.
Therefore, the resulting angle is given as a transformation of the ratio of Gaussian random variables.
Using the fact that $\avg{\X}>0$, the distribution of this ratio~${\ratioXY=\avg{\Y}/\avg{\X}}$ can be approximated with a normal distribution as~\cite{Hinkley1969}
\begin{equation}\label{eq:approx-dist-ratio-normal}
	\cdf_{\ratioXY}(w) = \cdfnormal\left(\frac{\numsteps\mean_x w}{\sqrt{\numsteps\left(w^2 \var_x+ \var_y\right)}}\right)\,.
\end{equation}
From this, the density of~$\ratioXY$ can then be determined as
\begin{equation}
	\pdf_{\ratioXY}(w) = \pdfnormal\left(\frac{\numsteps\mean_x w}{\sqrt{\numsteps\left(w^2 \var_x+ \var_y\right)}}\right) \frac{\sqrt{\numsteps}\mean_x \var_y}{\left(w \var_x+ \var_y\right)^{\frac{3}{2}}}\,.
\end{equation}
Finally, with the transformation~$\ratioXY=\tan\phaseresult_{\numsteps}$, we obtain the approximate \gls{cdf} and \gls{pdf} of~$\phaseresult_{\numsteps}$ as
\begin{align}
	\cdf_{\phaseresult_{\numsteps}}(\theta) &= \cdf_{\ratioXY}(\tan \theta)\\
	\pdf_{\phaseresult_{\numsteps}}(\theta) &= \frac{\pdf_{\ratioXY}(\tan \theta)}{\cos^2 \theta}
\end{align}
respectively.

\begin{example}[Approximate Distribution of Angle~$\phaseresult_{\numsteps}$]\label{ex:angle-n-large}
	A numerical example of the distribution of the angle for different values of~$\numsteps$ and~$\maxangle$ is shown in \autoref{fig:pdf-angle-large-n}.
	In particular, we compare the approximation for large~$\numsteps$ from~\eqref{eq:result-pdf-phase-large-n} with a \gls{pdf} determined from \gls{mc} simulations with $10^7$~samples.
	While the approximation is not as good for~${\numsteps=5}$ in \autoref{fig:pdf-angle-n5}, it already is very close.
	As expected, this improves further for a larger value of~${\numsteps=30}$ shown in \autoref{fig:pdf-angle-n30}.
	From this example and \autoref{ex:radius-n-large}, it can be concluded that the approximation for large~$\numsteps$ yields a better performance for the angle~$\phaseresult_{\numsteps}$ compared to the radius~$\radius_{\numsteps}$ for smaller values of~$\numsteps$.

	\begin{figure}
		\centering
		\subfloat[{$\numsteps=5$\label{fig:pdf-angle-n5}}]{\input{img/pdf-angle-n5.tex}}

		\subfloat[{$\numsteps=30$\label{fig:pdf-angle-n30}}]{\input{img/pdf-angle-n30.tex}}
		\caption{%
			Approximation of the density of the angle~$\phaseresult_{\numsteps}$ for large~$\numsteps$ from~\eqref{eq:result-pdf-phase-large-n} together with a histogram obtained through \gls{mc} simulations with $10^7$~samples.
			(\autoref{ex:angle-n-large})
		}
		\label{fig:pdf-angle-large-n}
	\end{figure}
\end{example}

%% file: img/pdf-angle-n5.tex
\begin{tikzpicture}%
	\begin{axis}[
		betterplot,
		height=.23\textheight,
		xlabel={Angle~$\phaseresult_5$},
		ylabel={\Gls{pdf}~$\pdf_{\phaseresult_5}$},
		xmin=-1.25,
		xmax=1.25,
		ymin=0,
		ymax=3.2,
		legend style={
			font=\small,
			align=left,
		},
		]
		\pgfplotstableread{data/results-approx-n-large-a0.500-n5.dat}\ta
		\pgfplotstableread{data/results-mc-n-large-a0.500-n5.dat}\tam

		\pgfplotstableread{data/results-approx-n-large-a1.250-n5.dat}\tb
		\pgfplotstableread{data/results-mc-n-large-a1.250-n5.dat}\tbm

		\addplot+[mark repeat=30] table[x=angle,y=pdfAngle] {\ta};
		\addlegendentry{Approx.\\$\maxangle=0.5$}

		\addplot+[mark repeat=3,const plot mark right] table[x=angle,y=histAngle] {\tam};
		\addlegendentry{MC\\$\maxangle=0.5$}

		\addplot+[mark repeat=30] table[x=angle,y=pdfAngle] {\tb};
		\addlegendentry{Approx.\\$\maxangle=1.25$}

		\addplot+[mark repeat=3,const plot mark right] table[x=angle,y=histAngle] {\tbm};
		\addlegendentry{MC\\$\maxangle=1.25$}
	\end{axis}
\end{tikzpicture}

%% file: img/pdf-angle-n30.tex
\begin{tikzpicture}%
	\begin{axis}[
		betterplot,
		height=.23\textheight,
		xlabel={Angle~$\phaseresult_{30}$},
		ylabel={\Gls{pdf}~$\pdf_{\phaseresult_{30}}$},
		xmin=-1.25,
		xmax=1.25,
		ymin=0,
		ymax=7.5,
		legend style={
			font=\small,
			align=left,
		},
		]
		\pgfplotstableread{data/results-approx-n-large-a0.500-n30.dat}\ta
		\pgfplotstableread{data/results-mc-n-large-a0.500-n30.dat}\tam

		\pgfplotstableread{data/results-approx-n-large-a1.250-n30.dat}\tb
		\pgfplotstableread{data/results-mc-n-large-a1.250-n30.dat}\tbm

		\addplot+[mark repeat=30] table[x=angle,y=pdfAngle] {\ta};
		\addlegendentry{Approx.\\$\maxangle=0.5$}

		\addplot+[mark repeat=3,const plot mark right] table[x=angle,y=histAngle] {\tam};
		\addlegendentry{MC\\$\maxangle=0.5$}

		\addplot+[mark repeat=30] table[x=angle,y=pdfAngle] {\tb};
		\addlegendentry{Approx.\\$\maxangle=1.25$}

		\addplot+[mark repeat=3,const plot mark right] table[x=angle,y=histAngle] {\tbm};
		\addlegendentry{MC\\$\maxangle=1.25$}
	\end{axis}
\end{tikzpicture}

%% file: large-n-joint.tex
\subsection{Joint Distribution of Radius and Angle~\texorpdfstring{$(\radius_{\numsteps}, \phaseresult_{\numsteps})$}{(Rₙ, θₙ)}}

The application of the \gls{clt} to the real and imaginary parts shows that they become jointly Gaussian distributed for large~$\numsteps$, cf.~\eqref{eq:clt-x-y-components}.
From this, we use the transformation to/from polar coordinates according to~\eqref{eq:radius-large-n-avg-xy} and~\eqref{eq:angle-large-n-avg-xy} to obtain the joint density of the radius~$\radius_{\numsteps}$ and angle~$\phaseresult_{\numsteps}$ as
\begin{align}
	\pdf_{\radius_{\numsteps}, \phaseresult_{\numsteps}}(r, \theta)
	&= r\pdf_{\numsteps\avg{\X},\numsteps\avg{\Y}}\left(r\cos\theta, r\sin\theta\right)\\
	&= \frac{r}{\numsteps^2}\pdf_{\avg{\X}}\left(\frac{r\cos\theta}{\numsteps}\right) \pdf_{\avg{\Y}}\left(\frac{r\sin\theta}{\numsteps}\right)\\
	&= \frac{r}{\numsteps^2\std_x\std_y}\pdfnormal\left(\frac{r\cos\theta-\numsteps\mean_x}{\numsteps\std_x}\right) \pdfnormal\left(\frac{r\sin\theta}{\numsteps\std_y}\right),
	\label{eq:joint-pdf-large-n-deriv}
\end{align}
where~$\pdfnormal$ represents the \gls{pdf} of the standard normal distribution.

\begin{rem}[{Truncating the Distribution}]
	Since the approximation of the joint distribution in~\eqref{eq:joint-pdf-large-n-deriv} is based on the normal distribution, its support is~$\reals^2$.
	However, from \autoref{sub:n3-support}, we know that the exact support is restricted to a subset, which is bounded by the curves given in \autoref{thm:support-bound-characterization}.
	Therefore, we can improve the approximation by truncating it to the actual support, i.e., setting the \gls{pdf} to zero outside the exact support.
	As a numerical illustration, we compute this truncated probability for the values in \autoref{ex:radius-n-large}.
	For the small number of steps~${\numsteps=5}$, the values are~\num{0.053} and~\num{0.042} for~${\maxangle=0.5}$ and~${\maxangle=1.25}$, respectively.
	In contrast, for~${\numsteps=30}$, the values are~\num{8.7e-07} and~\num{1.9e-07}, respectively.
	As expected, the approximation is significantly better for larger~$\numsteps$, cf.~\autoref{fig:cdf-radius-large-n}.
\end{rem}

\begin{example}[{Approximate Joint Distribution of Radius~$\radius_{\numsteps}$ and Angle~$\phaseresult_{\numsteps}$}]\label{ex:joint-n-large}
	A numerical example of the approximate joint distribution is given in \autoref{fig:pdf-joint-n-large}, where we show the joint \gls{pdf} from~\eqref{eq:joint-pdf-large-n-deriv} for~${\numsteps=30}$ and~${\maxangle=0.5}$.
	Additionally, the exact boundary of the support from~\eqref{eq:result-support-inner} is highlighted.

	First, it can be seen that the probability mass of the approximation is already small outside the actual support. %
	Next, it can be observed that the probability mass is concentrated around a radius of around~${\radius_{30}=28.8}$ and an angle around~${\phaseresult_{30}=0}$.
	This is different compared to the distribution for small~$\numsteps$, e.g., the two-step distribution in \autoref{fig:pdf-joint-n2}, where the probability mass is concentrated around~${\radius_{\numsteps}=\numsteps}$.
	A similar shift in distribution can be observed in the traditional case (${\maxangle=\pi}$), where the probability mass is also more concentrated close to~${\radius_{\numsteps}=\numsteps}$ for small~$\numsteps$, while the distribution of the radius shifts to a Rayleigh distribution for large~$\numsteps$~\cite{Rayleigh1905}.

	\begin{figure}
		\centering
		\input{img/pdf-joint-n30.tex}
		\caption{%
			Approximation of the joint density of the radius~$\radius_{\numsteps}$ and $\phaseresult_{\numsteps}$ from~\eqref{eq:result-pdf-joint-large-n} for~${\numsteps=30}$ and~${\maxangle=0.5}$.
			Additionally, the exact boundary of the support from~\eqref{eq:result-support-inner} is shown.
			(\autoref{ex:joint-n-large})
		}
		\label{fig:pdf-joint-n-large}
	\end{figure}
\end{example}

\begin{rem}[Intuition for Increasing~$\numsteps$]\label{rem:intuition-increasing-n}
	The derived results can also be used to verify some intuitions for an increasing number of steps~$\numsteps$.
	First, to get a resulting angle~$\phaseresult_{\numsteps}$ close to the boundary~$\maxangle$, all steps must be taken in this direction.
	However, as the number of steps increases, it becomes less likely that all steps line up in this direction, and the tail probabilities of~$\phaseresult_{\numsteps}$ should decrease with~$\numsteps$, which can be seen in \autoref{fig:pdf-angle-large-n}.
	Next, large resulting radii~$\radius_{\numsteps}$ only occur if all steps line up in the same direction.
	Intuitively, this translates to narrower support for extreme angles~${\phaseresult_{\numsteps}\approx\maxangle}$, since all steps line up in this case.
	This intuition is verified mathematically in \autoref{sub:n3-support} and illustrated in \autoref{fig:pdf-joint-n-large}.
	The concentration effect mentioned in \autoref{ex:joint-n-large}, which is due to the balancing effect of steps in various directions, is also closely related to this.
\end{rem}

%% file: img/pdf-joint-n30.tex
\begin{tikzpicture}
	\begin{axis}[
		width=.8\linewidth,
		height=.25\textheight,
		trig format plots=rad,
		view={0}{90},
		xmin=-.5,
		xmax=.5,
		colormap name=viridis,
		ylabel={Radius~$\radius_{30}$},
		xlabel={Angle~$\phaseresult_{30}$},
		axis on top,
		]
		\pgfplotstableread{data/results-joint-n-large-a0.500-n30.dat}\tbl

		\addplot3[surf,faceted color=none] table [x=angle,y=radius,z=pdfJoint]{\tbl};

		\addplot[plot1,very thick] table[x=angleInner,y=radInner]{data/results-support-a0.500-n30.dat};

	\end{axis}
\end{tikzpicture}

%% file: application-example.tex
\section{Application Example: Reducing Unnecessary Synchronization in Over-the-Air Computation}\label{sec:application-example}

After deriving the general results for the distribution of a two-dimensional random walk with restricted angles, we now further illustrate their application by connecting them to a problem from the field of over-the-air computation.
In particular, consider a system with $\numsteps$~devices that transmit simultaneously.
The fusion center receives the superposition of their individual signals. %
Due to propagation channels, the transmitted signal~${x_i\in\complexes}$ of user~$i$ experiences attenuation~${g_i\in\reals}$ and a phase offset~${\psi_i\in[-\pi, \pi]}$, which are assumed to be constant for many time slots, i.e., slow fading.
Therefore, the received signal~$r$ in time slot~$t$ is given as
\begin{equation*}
	r[t] = \sum_{i=1}^{\numsteps} x_i[t] g_i \exp\left(\imag\psi_i\right).
\end{equation*}
The transmitted signals~$x_i$ are composed of the data symbol~${s_i\in\reals}$, a power factor~${P_i\in\reals}$, and a phase shift~${\phi_i\in[-\pi,\pi]}$, i.e., we have ${x_i[t] = s_i[t] P_i \exp\left(\imag\phi_i\right)}$.
The goal of the fusion center is to estimate the sum of the data symbols, i.e.,~$\sum_{i=1}^{\numsteps}{s_i}$.

It is clear that the optimal solution is for each device to compensate for its channel effects by setting~${P_i=1/g_i}$ and~${\phi_i=-\psi_i}$, which reduces the received signal~$r$ to the sum of the data symbols.
For this, a calibration of all devices is done initially.
However, over time, the oscillators of the devices drift, causing an increase in the phase misalignment between~$\phi_i$ and $\psi_i$, i.e., the received signal becomes
\begin{equation}\label{eq:ota-receive-offset}
	r[t] = \sum_{i=1}^{\numsteps} s_i[t] \exp\left(\imag\phase_i\right),
\end{equation}
where~$\phase_i$ is the (unknown) phase drift of device~$i$ given by~${\phi_i=-\psi_i+\phase_i}$~\cite{Dahl2024ciss}.
While the estimation at the fusion center remains accurate when all~$\phase_i$ are small, the devices should be resynchronized when the phase misalignment becomes too large.
However, resynchronization is costly and should only be done if necessary.
Therefore, we periodically transmit pilot signals to detect when a resynchronization should be initiated.
We use the results derived in this work for this decision, as shown in the following.

By setting the pilot signals to~$s_i=1$ for all devices, we obtain the received signal from~\eqref{eq:ota-receive-offset} in the form of~\eqref{eq:def-random-walk}.
In this case, the random angles~$\phase_i$ correspond to the random phase misalignment of the devices.
The maximum angle~$\maxangle$ now corresponds to an application-specific threshold for tolerated phase misalignment.

If all devices are still within the tolerated misalignment~$\maxangle$, the received signal~$r$ is distributed like~$\Z_\numsteps$ from the previous sections.
Therefore, using the results for the \gls{pdf} from the previous sections, we can calculate the likelihood of the received symbol.
In particular, if we receive a symbol outside the support, which is derived in \autoref{sub:n3-support}, we can be sure that at least one device is outside the threshold.

\begin{example}[Detecting Necessary Resynchronization]
Consider the over-the-air computation system described above with ${\numsteps=30}$~devices.
The tolerated threshold for phase misalignment at each device is~${\maxangle=0.5}$.
Note that these values correspond to the ones in \autoref{ex:joint-n-large}, and an illustration of the \gls{pdf} and the support of the received signal can therefore be found in \autoref{fig:pdf-joint-n-large}.
If the fusion center receives the symbol~${r=28.5\cdot\exp(\imag \cdot 0.1)}$, it is likely that all devices are still within the tolerated misalignment range, so resynchronization is not necessary.
The likelihood of this point is around~\num{0.43} according to~\eqref{eq:result-pdf-joint-large-n}.

In contrast, if the received symbol is~${r=27\cdot\exp(-\imag\cdot 0.3)}$, it lies outside the support of~$\Z_{30}$ according to \autoref{thm:support-bound-characterization}.
Therefore, this symbol cannot be received if all devices are still within the tolerated range.
Hence, the phase offset is too large for at least one of the devices, making resynchronization necessary.
\end{example}

While points outside the support allow for the detection of necessary resynchronization with certainty, in practice, a system designer selects a likelihood threshold below which resynchronization is initiated.
This threshold is based on the tolerated probability of initiating a resynchronization when all devices are within the tolerated misalignment range.
The following example illustrates this.

\begin{example}[Probability of Unnecessary Resynchronization]\label{ex:ota-false-alarm}
	Consider an over-the-air computation system with $\numsteps$~devices that uses the likelihood threshold method to detect necessary resynchronizations, as described above.
	Specifically, a resynchronization is initiated if the {(log\nobreakdash-)likelihood} of the received symbol~$r$ is below a predefined threshold~$\gamma$.
	However, unnecessary resynchronizations can occur if we decide for resynchronization while all devices are still within the tolerated misalignment~$\maxangle$.
	The corresponding probability~$\probfa$ is therefore given by
	\begin{equation}
		\probfa = \Pr\left(\log \pdf_{\radius_{\numsteps}, \phaseresult_{\numsteps}; \maxangle}(r) \leq \gamma\right) \,,
	\end{equation}
	where $\pdf_{\radius_{\numsteps}, \phaseresult_{\numsteps}; \maxangle}$ is the \gls{pdf} derived in the previous sections with maximum angle~$\maxangle$.

	\autoref{fig:results-example-ota-pfa-n2} shows numerical results of~$\probfa$ for ${\numsteps=2}$ and different levels of tolerated misalignment~$\maxangle$ obtained by \gls{mc} simulations with $10^7$~samples.
	Given application-specific tolerances, system designers can use these results to select an appropriate threshold~$\gamma$.
	For example, if the application tolerates a maximum misalignment of ${\maxangle=0.5}$ and a probability of unnecessary resynchronization~${\probfa=0.2}$, the threshold~$\gamma$ should be set to around~${\gamma=2}$.
	For comparison, the same results are shown for a larger number of devices, namely~${\numsteps=30}$, in \autoref{fig:results-example-ota-pfa-n30}.
	As expected, the general trend is the same for both~$N$.
	In both cases, $\probfa$ increases with the threshold~$\gamma$ and also increases with the maximum tolerated misalignment~$\maxangle$.
	However, for large~$\numsteps$, we can notice the aforementioned concentration effect. %
	This leads to steeper curves in \autoref{fig:results-example-ota-pfa-n30} compared to \autoref{fig:results-example-ota-pfa-n2}.

	\begin{figure}
		\centering
		\subfloat[{${\numsteps=2}$\label{fig:results-example-ota-pfa-n2}}]{\input{img/results-example-ota-pfa.tex}}

		\subfloat[{${\numsteps=30}$\label{fig:results-example-ota-pfa-n30}}]{\input{img/results-example-ota-pfa-n30.tex}}
		\caption{%
			Probability of unnecessary resynchronization for different tolerated misalignment levels~$\maxangle$.
			The results are obtained through \gls{mc} simulations with $10^7$ samples.
			(\autoref{ex:ota-false-alarm})}
		\label{fig:results-example-ota-pfa}
	\end{figure}
\end{example}

\begin{rem}[{Typical Values of~$\numsteps$\label{rem:typical-num-steps}}]
	The advantage of over-the-air computation over consecutive transmissions increases with the number of devices~$\numsteps$, since only a single time slot is needed for data aggregation instead of~$\numsteps$.
	This beneficial effect of maximizing the number of devices has also been observed for specific applications of over-the-air computation, such as federated learning~\cite{Yang2020}.
	The typical values of~$\numsteps$ considered in literature are around~\numrange{20}{40}~\cite{Cai2018,Yang2020} or even larger~\cite{Liu2020}.
	Therefore, in applications in the context of over-the-air computation, the approximation for large~$\numsteps$ from \autoref{sec:large-n} is most practical.
	This is also consistent with the numerical results in \autoref{ex:radius-n-large} and \autoref{ex:angle-n-large}, where the approximation is already very accurate for ${\numsteps=30}$.
\end{rem}

%% file: img/results-example-ota-pfa.tex
\begin{tikzpicture}%
	\begin{axis}[
		betterplot,
		height=.23\textheight,
		xlabel={Log-likelihood Threshold~$\gamma$},
		ylabel={Prob. Unnec. Resync.~$\probfa$},
		xmin=-1,
		xmax=10,
		ymin=0,
		ymax=1,
		legend style={
			font=\small,
			align=left,
			anchor=south east,
			at={(.97,.1)},
		},
		]
		\pgfplotstableread{data/results-example-ota.dat}\ta

		\addplot+[mark repeat=20] table[x=gamma,y={a0.10}] {\ta};
		\addlegendentry{$\maxangle=0.1$}

		\addplot+[mark repeat=20] table[x=gamma,y={a0.50}] {\ta};
		\addlegendentry{$\maxangle=0.5$}

		\addplot+[mark repeat=20] table[x=gamma,y={a0.79}] {\ta};
		\addlegendentry{$\maxangle=\pi/4$}
	\end{axis}
\end{tikzpicture}

%% file: img/results-example-ota-pfa-n30.tex
\begin{tikzpicture}%
	\begin{axis}[
		betterplot,
		height=.23\textheight,
		xlabel={Log-likelihood Threshold~$\gamma$},
		ylabel={Prob. Unnec. Resync.~$\probfa$},
		xmin=-4,
		xmax=5,
		ymin=1e-7,
		ymax=1,
		ymode=log,
		legend style={
			font=\small,
			align=left,
			anchor=south east,
			at={(.8,.1)},
		},
		]
		\pgfplotstableread{data/results-example-ota-N30.dat}\ta

		\addplot+[mark repeat=20] table[x=gamma,y={a0.10}] {\ta};
		\addlegendentry{$\maxangle=0.1$}

		\addplot+[mark repeat=20] table[x=gamma,y={a0.50}] {\ta};
		\addlegendentry{$\maxangle=0.5$}

		\addplot+[mark repeat=20] table[x=gamma,y={a0.79}] {\ta};
		\addlegendentry{$\maxangle=\pi/4$}
	\end{axis}
\end{tikzpicture}

%% file: conclusion.tex
\section{Conclusion}\label{sec:conclusion}
In this work, we have derived the joint and marginal distributions of a two-dimensional random walk in which the angles of each step are uniformly distributed over the subset of the full circle.
The results in this paper can act as a reference for future work encountering this or a closely related problem.
In particular, these results will be beneficial for work on over-the-air computation as the example in \autoref{sec:application-example} has shown.